%% file: Elimination-FOL_full.tex
\documentclass[11pt,a4paper]{article}

\usepackage{amsmath,amsfonts,amssymb,amsthm,wrapfig}
\usepackage{mathbbol}
\usepackage{amsthm}
\usepackage{graphicx,color,colortbl}
\usepackage{boxedminipage}
\usepackage[linesnumbered, boxed, algosection,noline]{algorithm2e}
\usepackage{framed}
\usepackage{thmtools}
\usepackage{thm-restate}
\usepackage{xspace}
\usepackage{vmargin}
\setmarginsrb{1in}{1in}{1in}{1in}{0pt}{0pt}{0pt}{6mm}
\usepackage{todonotes}
 \usepackage[pdftex, plainpages = false, pdfpagelabels, 
                 bookmarks=false,
                 bookmarksopen = true,
                 bookmarksnumbered = true,
                 breaklinks = true,
                 linktocpage,
                 pagebackref,
                 colorlinks = true,  
                 linkcolor = blue,
                 urlcolor  = blue,
                 citecolor = red,
                 anchorcolor = green,
                 hyperindex = true,
                 hyperfigures
                 ]{hyperref} 
 \usepackage{xifthen}
 \usepackage{tabularx}
 \usepackage{tcolorbox} 
 \usetikzlibrary{calc}

\DeclareMathOperator{\operatorClassNP}{{\sf NP}}
\newcommand{\classNP}{\ensuremath{\operatorClassNP}\xspace}
\DeclareMathOperator{\operatorClassCoNP}{{\sf coNP}}
\newcommand{\classCoNP}{\ensuremath{\operatorClassCoNP}\xspace}
\DeclareMathOperator{\operatorClassFPT}{{\sf FPT}\xspace}
\newcommand{\classFPT}{\ensuremath{\operatorClassFPT}\xspace}
\DeclareMathOperator{\operatorClassW}{{\sf W}}
\newcommand{\classW}[1]{\ensuremath{\operatorClassW[#1]}\xspace}

\DeclareMathOperator{\operatorClassXP}{{\sf X}P\xspace}
\newcommand{\classXP}{\ensuremath{\operatorClassXP}\xspace}
 \DeclareMathOperator{\operatorClassPSPACE}{{\sf PSPACE}\xspace}
\newcommand{\classPSPACE}{\ensuremath{\operatorClassPSPACE}\xspace}

\newtheorem{theorem}{Theorem}
\newtheorem{lemma}{Lemma}

\newtheorem{definition}{Definition}
\newtheorem{observation}{Observation}
\newtheorem{proposition}{Proposition}

\newcommand{\yes}{{yes}}

\newcommand{\Oh}{\mathcal{O}}

\newcommand{\ed}{{\sf ed}}

\newcommand{\edphi}[1]{{\sf ed}_{\varphi}^{#1}}
\newcommand{\edphione}{\text{\sf conn}}
\newcommand{\edphitwo}{\text{\sf prop}}
\newcommand{\edphthree}{\text{\sf depth}}
\newcommand{\edphfour}{\text{\sf part}}
\newcommand{\property}{{\cal P}}

\newcommand{\dpt}{{\sf depth}}
\newcommand{\comp}{{\sf comp}}

\newcommand{\pname}{\textsc}
\newcommand{\ProblemFormat}[1]{\pname{#1}}
\newcommand{\ProblemIndex}[1]{\index{problem!\ProblemFormat{#1}}}
\newcommand{\ProblemName}[1]{\ProblemFormat{#1}\ProblemIndex{#1}{}\xspace}

 \newcommand{\probMC}{\ProblemName{Model Checking}}
 \newcommand{\probED}[1]{\ProblemName{Elimination Distance--($#1$) to $\varphi$}}
 \newcommand{\probDEL}{\ProblemName{Deletion to $\varphi$}}



\newlength{\RoundedBoxWidth}
\newsavebox{\GrayRoundedBox}
\newenvironment{GrayBox}[1]%
   {\setlength{\RoundedBoxWidth}{.93\textwidth}
    \def\boxheading{#1}
    \begin{lrbox}{\GrayRoundedBox}
       \begin{minipage}{\RoundedBoxWidth}}%
   {   \end{minipage}
    \end{lrbox}
    \begin{center}
    \begin{tikzpicture}%
       \node(Text)[draw=black!20,fill=white,rounded corners,%
             inner sep=2ex,text width=\RoundedBoxWidth]%
             {\usebox{\GrayRoundedBox}};
        \coordinate(x) at (current bounding box.north west);
        \node [draw=white,rectangle,inner sep=3pt,anchor=north west,fill=white] 
        at ($(x)+(6pt,.75em)$) {\boxheading};
    \end{tikzpicture}
    \end{center}}     

\newenvironment{defproblemx}[2][]{\noindent\ignorespaces%
                                \FrameSep=6pt%
                                \parindent=0pt%
                \vspace*{-1.5em}
                \ifthenelse{\isempty{#1}}{%
                  \begin{GrayBox}{\textsc{#2}}%
                }{%
                  \begin{GrayBox}{\textsc{#2}  parameterized by~{#1}}%
                }
                \begin{tabular*}{\textwidth}{@{\hspace{.1em}} >{\itshape} p{1.8cm} p{0.8\textwidth} @{}}%
            }{
                \end{tabular*}%
                \end{GrayBox}%
                \ignorespacesafterend
            }  

\newcommand{\defparproblema}[4]{
  \begin{defproblemx}[#3]{#1}
    Input:  & #2 \\
    Task: & #4
  \end{defproblemx}
}%

\title{Parameterized Complexity of  Elimination Distance to First-Order Logic Properties
\thanks{An extended abstract of this paper was accepted for
36th Annual Symposium on Logic in Computer Science (LICS 2021). 
The two first authors have been supported by the Research Council of Norway via the project ``MULTIVAL". The third author has been supported  by   the ANR projects DEMOGRAPH (ANR-16-CE40-0028) and ESIGMA (ANR-17-CE23-0010) and the French-German Collaboration ANR/DFG Project UTMA ANR-20-CE92-0027.}
}

\author{
Fedor V. Fomin\thanks{
Department of Informatics, University of Bergen, Norway, \textsf{\{fedor.fomin, petr.golovach\}@uib.no}} \addtocounter{footnote}{-1}
\and
Petr A. Golovach\footnotemark{}
\and 
Dimitrios M. Thilikos\thanks{AlGCo project team, CNRS, LIRMM, Université de Montpellier, Montpellier, France, \textsf{sedthilk@thilikos.info}}}
\date{}

\begin{document}
\maketitle

\begin{abstract}
\noindent The \emph{elimination distance} to some target graph property ${\cal P}$  is a general graph modification parameter introduced by Bulian and Dawar.  We initiate the study of elimination distances to graph properties expressible in first-order logic. We delimit the problem's fixed-parameter tractability by identifying sufficient and necessary conditions on the structure of prefixes of first-order logic formulas. Our main result is the following meta-theorem: For every graph property ${\cal P}$  expressible by a first order-logic formula $\varphi\in \Sigma_3$, that is,  of the form 
\[\varphi=\exists x_1\exists x_2\cdots \exists x_r\ \ \forall y_{1}\forall y_{2}\cdots \forall y_{s}\ \ \exists z_1\exists z_2\cdots \exists z_t~~ \psi,\]
where $\psi$ is a quantifier-free first-order formula, 
checking whether the elimination 
distance of a graph to ${\cal P}$ does not exceed $k$, is \emph{fixed-parameter tractable} parameterized by $k$.
Properties of graphs expressible by formulas from $\Sigma_3$ include being  of bounded degree, excluding a forbidden subgraph, or containing a bounded dominating set. 
We complement this theorem by showing that 
such a general statement does not hold 
 for  formulas with even slightly more expressive prefix structure: There are  formulas $\varphi\in \Pi_3$, for which computing elimination distance is ${\sf W}[2]$-hard. 
 \end{abstract}


\section{Introduction}\label{sec_intro} 
One of the successful concepts in parameterized complexity is the ``distance from triviality'' \cite{GuoHN04}. Roughly speaking, a parameter can measure the ``distance'' of the given instance from an instance that is solvable efficiently and then exploit such a distance algorithmically. In graph problems, a standard measure of distance from triviality is the {\sl vertex deletion distance}   to some  specific graph property $\property$. That is, the minimum number of vertices whose deletion  results in a graph in $\property$. 
 An interesting alternative to vertex deletion distance, called 
 {\sl elimination distance} was introduced by Bulian and Dawar \cite{BulianD16} in their study of the parameterized complexity of the graph isomorphism problem.
 The  \emph{elimination distance} of a graph $G$ to  graph property $\property$ is\begin{equation*}
\ed_{\property}(G)=
\begin{cases}
0,&\mbox{if }G\in\property,\\
1+\min_{v\in V(G)}\ed_{\property}(G-v), &\mbox{if }G\notin\property\text{ and }G\text{ is connected},\\
\max\{\ed_{\property}(C)\mid C\text{ is a component of }G\},&\mbox{otherwise}.
\end{cases}
\end{equation*}
Arguably, elimination distance can be seen as a non-deterministic version of vertex deletion distance, 
where the source of non-determinism is connectivity: each vertex removal creates connected components and the elimination should be applied to each one of them as an independent  vertex deletion scenario. In the most simple case where $\property$ is the property of being edgeless, vertex deletion distance to $\property$ generates {\sl vertex cover}, while the elimination distance to $\property$ generates  {\sl tree-depth}  \cite{NesetrilO0806b}. 

In their follow-up work, Bulian and Dawar \cite{BulianD17} proved that deciding whether
 a given $n$-vertex graph has elimination distance at most $k$ to any minor-closed property of graphs can be done by an algorithm running in time $f(k)\cdot n^{\Oh(1)}$ (that is an {\sf FPT}-algorithm), and thus is fixed-parameter tractable parameterized by $k$. 
 In the same paper, Bulian and Dawar \cite{BulianD17} asked whether computing the elimination distance to graphs of bounded degree is fixed-parameter tractable.

\paragraph{The problem.}
The question of Bulian and Dawar
  is the departure point of our study. 
Every graph property  
characterized by a finite set of forbidden induced 
 subgraphs (and thus the bounded degree property as well) is first-order logic definable (in short, FOL-definable), i.e., there is a FOL formula $\varphi$ 
where $\property=\{G\mid G\models \varphi\}$. It is well-known that \probMC for a FOL formula $\varphi$, that is deciding whether $G\models\varphi$, can be done in time $n^{\Oh(|\varphi|)}$. It is also 
easy to design an algorithm that, in time $n^{\Oh(k)}\cdot n^{\Oh(|\varphi|)}$,  decides whether the elimination distance to a property expressible by a FOL formula $\varphi$ is at most $k$. Thus, for every FOL 
formula $\varphi$, the problem asking, given as input a graph $G$ and a non-negative integer $k$, whether the elimination distance from $G$ to $\property_{\varphi}:=\{G\mid G\models \varphi\}$ is $k$ is in the parameterized complexity class {\sf XP} (when parameterized by $k$). This brings us to the following question. \\

 \begin{tcolorbox}[colback=green!5!white,colframe=blue!40!black] 
What is the parameterized complexity of computing the elimination distance to FOL-definable properties?
 \end{tcolorbox}\medskip

Notice that the above general question could also be made for higher order logic-definable properties.
In this direction, one may observe that there are formulas $\varphi$ in  existential second-order logic (ESOL) 
for which  \probMC  for  $\varphi$  is already intractable:  such  ESOL-definable 
problems  are  {\sc  Hamiltonian Cycle}
or {\sc 3-Coloring} that are {\sf NP}-complete. This means that for the corresponding ESOL formulas $\varphi$ 
the problem of checking whether $\ed_{\property_{\varphi}}(G)\leq k$, parameterized by $k$, is {\sf para-NP}-complete. Motivated by this, we delimit our study to the framework of 
first-order logic 
where our parameterized problem is in {\sf XP} for {\sl every} FOL-formula. This permits us to 
set up the problem that we consider in this paper, that is to completely determine the {\sl prefix classes} 
of FOL that  demark the  parametric-tractability borders of elimination distance to FOL-definable properties (that is {\sf FPT} versus {\sf W}-hardness).

The above question has been inspired by the study of  Gottlob, Kolaitis, Schwentick in \cite{GottlobKS04}
who provided an analogous dichotomy result  ({\sf P} versus {\sf NP}-complete)
for   ESOL-formulas, in several graph-theoretic contexts. They 
identified  the set $\frak{F}$ of prefix classes of  ESOL such that,  
if $\varphi\in \frak{F}$, then   \probMC for  $\varphi$ is polynomially solvable, 
while every prefix class not in $\frak{F}$ contains some ESOL formula $\varphi$ where \probMC for  $\varphi$ is {\sf NP}-complete.

 \medskip\noindent\textbf{Our results.}
 We identify sufficient and necessary conditions on the structure of   prefixes of  first-order logic formulas demarcating  tractability borders for computing the elimination distance.

 Our main algorithmic contribution is the proof that computing the elimination distance to any graph property defined by a formula from  $\Sigma_3$ is fixed-parameter tractable.  We formally define {prefix} classes $\Pi_i$ and $\Sigma_i$ in the next section.  For the purpose  of this introduction, it is sufficient to know that 
every formula in $\varphi\in \Sigma_3$ can be written in the form 
\[\varphi=\exists x_1\exists x_2\cdots \exists x_r\ \ \forall y_{1}\forall y_{2}\cdots \forall y_{s}\ \ \exists z_1\exists z_2\cdots \exists z_t~~ \psi,\]
where $\psi$ is a  quantifier-free FOL-formula and $r,s,t$ are non-negative integers.

 Every graph property characterized by a finite set of forbidden subgraphs  can be expressed  by  $\varphi \in \Sigma_3$. Actually, for this partiuclar purpose, we can consider  even more restricted formulas   $\varphi \in \Pi_1\subset \Sigma_3$ with only  $\forall$ quantifications over variables. 
The property that the diameter of a graph is at most two cannot be expressed by using forbidden subgraphs but can easily be written as an  FOL formula from $\Pi_2\subset \Sigma_3$: $\forall u\forall v\exists w [(u=v)\vee (u\sim v)\vee ((u\sim w)\wedge(v\sim w))]$. Another interesting example of a property expressible in $\Sigma_3$ is the property of containing a universal vertex, and, more generally, having an $r$-dominating set of size at most $d$ for constants $r$ and $d$.   Having a twin-pair, that is a pair of vertices with equal neighborhoods,  is also the property expressible in $\Sigma_3$. 
 
\begin{theorem}[\textbf{Informal}]\label{thm_algorithm_informal}
For every  $\varphi\in\Sigma_3$, $n$-vertex graph $G$, and $k\geq 0$, deciding whether the elimination distance from $G$ to property $\property_\varphi$, is at most $k$, can be done in time $f(k)\cdot n^{\Oh(|\varphi|)}$ for some function $f$ of $k$ only. 
\end{theorem}
Our second theorem shows that the assumptions on the prefix of the formula   are necessary. Let  $\Pi_3$ be the class of   first-order logic formulas of the form 
\[\varphi=\forall x_1\forall x_2\cdots\forall x_r\ \  \exists y_{1} \exists  y_{2}\cdots \exists  y_{s}\ \ \forall z_1\forall z_2\cdots \forall z_t~~ \psi,\]
where $\psi$ is a  FOL-formula without quantifiers and $s,t,q$ are non-negative integers.
We show that 
 
 \begin{theorem}[\textbf{Informal}]\label{thm_W_hard_informal}
There is $\varphi\in\Pi_3$ such that deciding whether the elimination distance to $\property_\varphi$ is at most $k$,  is  \classW{2}-hard parameterized by $k$.
\end{theorem}
 
 \paragraph{Variants of elimination distance.}
 The main reason why we give  informal statements of our theorems in the introduction is due to the following issue.
 The definition of elimination distance is tailored to the graph properties $\property$ with the condition that $G\in \property$ if and only if $C\in \property $ for every component $C$ of $G$.  
 Graph properties defined by FOL do not necessarily satisfy such a condition. This leads to ambiguities. As an example, consider  graph property $\property=\{G\mid G\models\forall x\forall y~x=y \}$. Thus $G\in \property$ if and only if $G$ is  a single-vertex graph. Let $G$ be an edgeless graph with $n\geq 2$ vertices. 
 Since  $G\notin \property$ it would be a natural assumption that the elimination distance from $G$ to $\property$ is positive. However, it is not: 
  every connected component of $G$ is in $\property$  and, therefore, \[\ed_{\property}(G)=\max\{\ed_{\property}(C)\mid C\text{ is a component of }G\}=0.\] To avoid such ambiguities, we refine the definition of the elimination distance. 

Since we consider   graph properties $\property_\varphi=\{G\mid G\models \varphi\}$ for   formulas $\varphi$, we define the distances with respect to formulas. Notice that the notion of elimination distance combines ``connectivity'' and ``inclusion'' in a graph class. Depending on which of these two properties we want to 
prioritize, we give different definitions.  Let $\varphi$ be a FOL formula. 

\begin{definition}[Elimination distances $\edphi{\edphione}$ and $\edphi{\edphitwo}$]
The first definition prioritize on connectivity. For a graph $G$, we set 
\begin{equation*}
\edphi{\edphione}(G)=
\begin{cases}
0,& \mbox{if }G\models\varphi,\\
1+\min_{v\in V(G)}\edphi{\edphione}(G-v), &\mbox{otherwise},
\end{cases}
\end{equation*}
if $G$ is   connected. We set  
\begin{equation*}
\edphi{\edphione}(G)=\max\{\edphi{\edphione}(C)\mid C\text{ is a component of }G\}
\end{equation*}
when  $G$ is not connected. The second definition prioritize   on the graph property  
\begin{equation*}
\edphi{\edphitwo}(G)=
\begin{cases}
0,&\mbox{if }G\models\varphi,\\
1+\min_{v\in V(G)}\edphi{\edphitwo}(G-v), &\mbox{if }G\not\models\varphi\text{ and }G\text{ is connected},\\
\max\{1,\max\{\edphi{\edphitwo}(C)\mid C\text{ is a component of }G\}\},&\mbox{otherwise}.
\end{cases}
\end{equation*}
We assume that $\edphi{\edphione}(G)=\edphi{\edphitwo}(G)=0$ if $G=(\emptyset,\emptyset)$ for any formula $\varphi$. 
\end{definition}
If 
$\varphi$ is such that $G\in \property_{\varphi}$ if and only if $C\in \property_{\varphi}$ for every component $C$ of $G$, then $\ed_{\property_{\varphi}}(G)=\edphi{\edphione}(G)=\edphi{\edphitwo}(G)$. However, in general
  $\edphi{\edphione}(G)$ and $\edphi{\edphitwo}(G)$ may differ significantly.
Consider $\varphi=\exists u\exists v~\neg(u=v)\wedge \neg(u\sim v)$ that defines the property that a graph has two nonadjacent vertices. Let $G$ be the disjoint union of the complete $n$-vertex graph $K_n$ and an isolated vertex. Then $G\models\varphi$ and, therefore, $\edphi{\edphitwo}(G)=0$. On the other hand,  $G$ is not connected and it is easy to see that $\edphi{\edphione}(K_n)=n$.

  Given the above, a more precise statement of Theorems~\ref{thm_algorithm_informal} and \ref{thm_W_hard_informal}  is that they hold for both distances  $\edphi{\edphione} $ and $\edphi{\edphitwo}$. We also define another interesting type of elimination distance $\edphi{\edphthree}$ that focuses on the depth of the elimination set. We prove Theorems~\ref{thm_algorithm_informal} and \ref{thm_W_hard_informal} hold for that distance as well. Precise statements of both theorems and their proofs are given in Sections~\ref{sec_FPT} and \ref{sec_hard}.

 \medskip\noindent\textbf{Related work.} 
 Results of this work fit into two popular trends in logic and parameterized complexity. A 
significant amount of research in descriptive complexity is devoted to the study of prefix classes of certain logics. We refer to   the book of 
 B{\"o}rger,  Gr{\"a}del,   and Gurevich  
 \cite{borger2001classical},  as well as the aforementioned 
 work of 
  Gottlob, Kolaitis and Schwentick~\cite{DBLP:journals/jacm/GottlobKS04} for further references.
 The study of graph modification problems is a well-established trend in parameterized complexity. The books~\cite{CyganFKLMPPS15,DowneyF13,FlumG06,Niedermeierbook06} provide a comprehensive overview of the area.
 In particular, Fomin, Golovach, and Thilikos   \cite{FominGT20} 
 studied parameterized complexity of computing vertex deletion distance and edge editing to graph properties defined by first-order logic formulas; \cite[Theorem~1]{FominGT20}  establishes fixed-parameter tractability for  vertex removal to a graph property $\property_\varphi$ for $\varphi \in \Sigma_3$ and shows that the problem is \classW{2}-hard for some 
 $\varphi\in\Pi_3$. While our Theorem~\ref{thm_algorithm_informal} reaches the same tractability border for the elimination distance, the proof is significantly more complicated.

  The general question on the parameterized complexity of elimination distance to graph properties was settled by  Bulian and Dawar~\cite{BulianD16,BulianD17}. Properties 
  that have been considered so far are {\sl minor-free graph classes} \cite{BulianD17}, {\sl cluster graphs}  \cite{SridharanPASK21,agrawal_et_al:LIPIcs:2020:13304}, {\sl bounded degree graphs} \cite{BulianD16,DLindermayrSV20},  and $H$-free graphs~\cite{SridharanPASK21}. Moreover,   
 Hols et al. \cite{hols_et_al:LIPIcs:2020:11897} studied the existence of polynomial kernels for the {\sc Vertex Cover} problem parameterized by the size of a deletion set to graphs of bounded elimination distance to certain graph classes.
 Lindermayr, Siebertz, and  Vigny~\cite{DLindermayrSV20} proved that computing the elimination distance to graphs of bounded degree is fixed-parameter tractable when the input does not contain $K_5$ as a minor. 
 While preparing our paper, we have learned about the very recent work of   Agrawal et al.~\cite{SridharanPASK21}.
 Agrawal et al. established fixed-parameter tractability of  
 computing an elimination distance to any graph property characterized
by a finite set of graphs as forbidden induced subgraphs. Since graphs of bounded vertex degree can be characterized by a finite set of forbidden induced subgraphs, the work of  Agrawal et al.  answers the question of Bulian and Dawar~\cite{ BulianD17}
about the elimination distance to graphs of bounded degree.    
 
 Comparing  with the result of  Agrawal et al.~\cite{SridharanPASK21}, our Theorem~\ref{thm_algorithm_informal} is  more general. First, it provides the tractability of the elimination ordering to a strictly larger family of graph properties. Every graph property described by a finite set of forbidden induced subgraphs is also definable by a formula from $\Sigma_3$. However, properties like having a universal vertex or bounded diameter, which are expressible in $\Sigma_3$, cannot be described by forbidden subgraphs. Second, Theorem~\ref{thm_algorithm_informal} holds for three variants of the elimination distance: 
  $\edphi{\edphione} $,  $\edphi{\edphitwo}$, and  $\edphi{\edphthree}$. With this terminology, the result of Agrawal et al. is only about computing  $\edphi{\edphione}$. When it comes to the proof techniques, both  Theorem~\ref{thm_algorithm_informal} and   the result of Agrawal et al. use recursive understanding, which seems to be a very natural technique for approaching problems about elimination distances. However, the details are quite different. 
 To deal with various types of elimination distances and FOL formulas in a uniform way,  
we need different combinatorial characterizations of the distances via sets  of bounded elimination depths. Furthermore, while solving our problems on unbreakable graphs is done by recursive branching algorithms, similarly to Agrawal et al., we do it in 
different way that exploits the random separation technique  to deal with our more general FOL framework. Moreover,     the analysis of components of the graph obtained by the deletion of an elimination set for 
computing $\edphi{\edphthree}(G)$, and especially, $\edphi{\edphitwo}(G)$, is great deal more challenging. In particular, this is the reason why we apply the random separation technique contrary to the more straightforward  tools used by Agrawal et al.

 \medskip\noindent\textbf{Overview of the approach.}
The firsts two variants of the elimination distance that we examine are defined recursively using the containment in the graph class  ${\cal P}_{\varphi}$ as the base case. We start by providing equivalent formulations that are more suitable from the algorithmic perspective.  For this, we introduce the notion of {\sl elimination set of depth at most $d$} that is a set $X\subseteq V(G)$ that can be bijectively mapped to a rooted tree $T$ of depth $d$ expressing selection of  elimination vertices in  recursive calls.
We next prove that $\edphi{\edphione}(G)\leq k$ if and only if $G$ has an elimination set of $X$ depth at most $k-1$
such that $C\models \varphi$ for every component $C$ of $G-X$. Similar, however more technical,
equivalent formulations is given for  $\edphi{\edphitwo}$. All alternative definitions and the proofs of their equivalences to the recursive ones are gathered in Section \ref{sec_elim}.

The new definitions allows to certify a solution by a set $X$ of bounded elimination depth. However, the size of $X$ could be unbounded. Moreover, there could be many connected components of $G-X$ and the sizes of these components could be immense. We use the {\sl recursive understanding technique}, introduced by Chitnis et al. in~\cite{ChitnisCHPP16} to reduce the solution of the initial problem to a much more structured problem. In the reduced problem, we can safely assume that each \yes-instance is certified by an elimination set $X$ whose size, 
as well as the size of the union of all {\sl but one} of the components of $G- X$,  is also bounded by a function of $k$.

More precisely, by making use of recursive understanding, we can consider only instances that are $(p(k),k)$-unbreakable for some suitably chosen function $p$. Roughly speaking, a graph is  $(p(k),k)$-unbreakable when it has no separator of size at most $k$ that partitions the graph in two parts of size at least $p(k)+1$ each. The  application of recursive understanding uses the meta-algorithmic result of Lokshtanov et al.~\cite{LokshtanovR0Z18} 
and the fact that all variants of the elimination distance to $\property_\varphi$
are expressible in monadic-second order logic (MSOL) when $\phi$ is a formula in FOL (\autoref{lem_expr}).

The  $(p(k),k)$-unbreakability permits us to assume that $|X|\leq p(k)+k$. Moreover, exactly one connected component $C_X$ of  $G-X$, is \emph{big}, that is of size at least $p(k)+1$, and 
the size of $G-V(C_X)$ is  bounded by some function of $k$ (see \autoref{lem_es_bound}). 
Given that $C_X$ is the big component corresponding to a solution $X$, we also  
consider the set $S_{X}$ of the neighbours  of the  vertices of $C_X$ in $G$ 
and we  set $U_X=V(G)\setminus (V(C_X)\cup S_X)$. 
We show that $|S_X|\leq k$ and $|U_X|\leq p(k)$.

Our next step is to use the {\sl random separation technique}, introduced by Cai, Chan, and Chan in~\cite{ChitnisCHPP16}.
 We construct in {\sf FPT}-time a family $\mathcal{F}$ of at most $f(k)\cdot  \log n$  partitions $(R,B)$ of $V(G)$ to ``red'' and ``blue'' vertices such that for every elimination set $X$ corresponding to a potential solution, ${\cal F}$ 
 contains some $(R,B)$ where   $U_X\subseteq R$ and $S_X\subseteq B$. In our algorithm, we go over all these blue-red partitions   and, for each one of them, we  
check whether there exists an  elimination set $X$  (called {\em colorful elimination set}) where all vertices in $S_X$ are blue and all vertices in $U_{X}$ are red.  

The correct ``guess'' of the above red-blue partition permits us to design  a recursive procedure that solves the latter problem, i.e., finds a colourful  elimination set $X$. 
This procedure is different for each of the three versions of the problem and its variants are presented in Subsection \ref{sec_unbreak}. The key task here is to identify the big component $C_X$.
It  runs in {\sf FPT}-time and its correctness is based on the prefix structure of the formula $\phi$.

\paragraph{Organization of the paper.} In Section~\ref{sec_prelim} we provide the basic 
definitions of the concepts that we use in this paper:  complexity classes graphs, 
and formulas. In Section~\ref{sec_elim} we prove some properties and relations between the 
elimination ordering variants that we consider. We also  provide alternative definitions
and we prove their equivalencies with the original ones.  The main algorithmic result 
is in Section~\ref{sec_FPT} where we explain how we apply the recursive understanding technique, the randomized separtion technique, and we present the branching procedure for the ``colourful version'' of each variant. Section~\ref{sec_hard} gives the lower bound of the paper. This uses parameterized reduction from the {\sc Set Cover} problem.
Finally, in Section~\ref{sec_concl},
we  provide some discussion on the kernelization complexity of our problems
as well as some directions on further research on elimination distance problems.

 \section{Preliminaries}\label{sec_prelim}

\paragraph{Sets.} 
We use $\mathbb{N}$ to denote the set of all non-negative numbers.
We denote by $\mathbf{a}=\langle a_1,\ldots,a_r\rangle$ a sequence of elements of a set $A$ and call $\mathbf{a}$ an \emph{$r$-tuple} of simply a \emph{tuple}. 
Note that the elements of $\mathbf{a}$ are not necessarily distinct. 
We denote by $\mathbf{a}\mathbf{b}=\langle a_1,\dots,a_r,b_1,\ldots,r_s\rangle$ the \emph{concatenation} of tuples $\mathbf{a}=\langle a_1,\ldots,a_r\rangle$ and $\mathbf{b}=\langle b_1\ldots,b_s\rangle$.

\paragraph{Parameterized Complexity.}  We refer to the recent books of Cygan et al.~\cite{CyganFKLMPPS15} and Downey and Fellows~\cite{DowneyF13} for the detailed introduction to the field. 
Formally, a \emph{parmeterized problem} is a language $L\subseteq\Sigma^*\times\mathbb{N}$, where $\Sigma^*$ is a set of strings over a finite alphabet $\Sigma$. This means that an input of a parameterized problem is a pair $(x,k)$, where $x$ is a string over $\Sigma$ and $k\in \mathbb{N}$ is a \emph{parameter}. 
A parameterized problem is \emph{fixed-parameter tractable} (or \classFPT) if it can be solved in time $f(k)\cdot |x|^{\Oh(1)}$ for some computable function~$f$. Also, we say that a parameterized problem belongs in the class $\classXP$  if it can be solved in time $|x|^{f(k)}$ for some computable function~$f$. 
The complexity class \classFPT contains  all fixed-parameter tractable problems.
Parameterized complexity theory also provides tools to disprove the existence of an \classFPT-algorithm for a problem under plausible complexity-theoretic assumptions. 
The standard way is to show that the problem is \classW{1} or \classW{2}-hard using a \emph{parameterized reduction} from a known \classW{1} or \classW{2}-hard problem; we refer to \cite{CyganFKLMPPS15,DowneyF13} for the formal definitions of the classes \classW{1} and \classW{2} and parameterized reductions. 

 \paragraph{Graphs.}
We consider only undirected simple graphs and use the standard graph theoretic terminology (see, e.g., \cite{Diestel12}). Throughout the paper we use $n$ to denote $|V(G)|$ if it does not create confusion.
For a set of vertices $S\subseteq V(G)$, we denote by $G[S]$ the subgraph of $G$ induced by the vertices from $S$. We also define $G-S=G[V(G)\setminus S]$; we write $G-v$ instead of $G-\{v\}$ for a single vertex set. For a vertex $v$, $N_G(v)$ denotes the \emph{open neighborhood} of $v$, that is, the set of vertices adjacent to $v$, and $N_G[v]=\{v\}\cup N_G(v)$ is the \emph{closed neighborhood}. For $S\subseteq V(G)$, $N_G(S)=\big(\bigcup_{v\in S}N_G(v)\big)\setminus S$ and $N_G[S]=\bigcup_{v\in S}N_G[v]$.
For a vertex $v$, $d_G(v)=|N_G(v)|$ denotes the \emph{degree} of $v$. 
A graph $G$ is \emph{connected} if for every two vertices $u$ and $v$, $G$ contains a path whose end-vertices are $u$ and $v$. 
For a positive integer $k$, $G$ is \emph{$k$-connected} if $|V(G)|\geq k$ and $G-S$ is connected for every $S\subseteq V(G)$ of size at most $k-1$.
A \emph{connected component} (or simply a \emph{component}) is an inclusion maximal induced connected subgraph of $G$. 
For two distinct vertices $u$ and $v$ of a graph $G$, a set $S\subseteq V(G)$ is a \emph{$(u,v)$-separator} if $G-S$ has no $(u,v)$-path. 

A \emph{rooted tree} is a tree $T$ with a selected node (we use the term ``node'' instead of ``vertex'' for such a tree) $r$ called a \emph{root}. The selection of $r$ defines the standard parent--child relation on $V(T)$. Nodes without children are called \emph{leaves} and we use $L(T)$ to denote the set of leaves of $T$.  The \emph{depth} $\dpt_T(v)$ of a node $v$ is the distance between $r$ and $v$, and the \emph{depth} (or \emph{height}) $\dpt(T)$ of $T$ is the maximum depth of a node. The nodes of the  $(r,v)$-path are called \emph{ancestors} of $v$.   We use $A_T(v)$ to denote the set of ancestors of $v$ in $T$. Note that $v$ is its own ancestor; we say that an ancestor is \emph{proper} if it is distinct from $v$.
Two nodes $u$ and $v$ of $T$ are \emph{comparable} if either $v$ is an ancestor of $u$ or $u$ is an ancestor of $v$. Otherwise, $u$ and $v$ are \emph{incomparable}.  A node $w$ of $T$ is the \emph{lowest common ancestor} of nodes $u$ and $v$ if $w$ is the ancestor of maximum depth of both $u$ and $v$. Note that the lowest common ancestor is unique and if $u$ and $v$ are incomparable then the lowest common ancestor is distinct from $u$ and $v$. 
A node $v$ is a \emph{descendant} of $u$ if $u$ if $u$ is an ancestor of $v$. By $D_T(u)$ we denote the set of descendants of $u$ in $T$. As with ancestors, a node   is its own descendant and we say that a descendant $v$ of $u$ is \emph{proper} if $u\neq v$. 
For a node $v$, the subtree induced by the descendants of $v$ is the \emph{subtree rooted in $v$}.

\paragraph{Formulas.}
In this paper we deal with \emph{first-order} and  \emph{monadic second-order} logic formulas on graphs.  

\medskip
The syntax of the first-order logic (FOL) formulas on graphs includes the logical connectives 
$\vee$, $\wedge$, $\neg$, variables for vertices, 
the quantifiers $\forall$, $\exists$ that are applied to these variables, the predicate $u \sim v$, where $u$ and $v$ are vertex variables and the interpretation is that $u$ and $v$ are adjacent, and the  equality of variables representing vertices. 
It also convenient to assume that we have the logical connectives $\rightarrow$ and $\leftrightarrow$.
An FOL formula $\varphi$ is  in {\em prenex normal form} if it is written as $\varphi={\tt Q}_{1}x_{1}{\tt Q}_{2}x_{2}\cdots{\tt Q}_{t}x_{t} \chi$  where 
each ${\tt Q}_i\in\{\forall,\exists\}$ is a quantifier,  $x_i$ is a variable, and $\chi$ is a quantifier-free part that depends on the variables $x_1,\ldots,x_t$. Then
${\tt Q}_{1}x_{1}{\tt Q}_{2}x_{2}\cdots{\tt Q}_{t}x_{t}$ is
 referred as the \emph{prefix} of $\varphi$. 
From now on, when we write ``FOL formula'', we mean an FOL formula on graphs that is in prenex normal form. Also we assume that a formula has no \emph{free}, that is, non-quantified variables unless we explicitly say that free variables are permitted.     
For an FOL formula $\varphi$  
and a graph $G$, we write $G\models \varphi$ to denote that $\varphi$ evaluates to \emph{true} on $G$.

We use the  arithmetic  hierarchy (also known as  Kleene-Mostowski  hierarchy) for the classification of formulas in the  first-order arithmetic language (see, e.g., \cite{Smorynski77}). For this, we define prefix classes    according to alternations of quantifiers, that is, switchings from $\forall$ to $\exists$ or vice versa in the prefix string of the formula. Here we allow a formula to have 
free variables.
Let $\Sigma_0=\Pi_0$  be the classes of FOL-formulas without quantifiers. 
For a positive integer $\ell$, the class $\Sigma_\ell$ contains formulas that may be written in the form 
\[\varphi=\exists x_1\exists x_2\cdots \exists x_s~ \psi,\]
where $\psi$ is  a $\Pi_{\ell-1}$-formula,  $s$ is some integer,  and $x_1,\ldots,x_s$ are free variables of $\psi$. Respectively, $\Pi_\ell$ consists of formulas 
\[\varphi=\forall x_1\forall x_2\cdots \forall x_s~ \psi,\]
where $\psi$ is  a $\Sigma_{\ell-1}$-formula and $x_1,\ldots,x_s$ are free variables of $\psi$.
Note that for $\ell'>\ell$, $\Sigma_{\ell}\cup\Pi_\ell\subseteq \Sigma_{\ell'}\cap\Pi_{\ell'}$, that is, every $\Sigma_\ell$ or $\Pi_\ell$ formula is both a 
$\Sigma_{\ell'}$ and $\Sigma_{\ell'}$-formula.

For technical reasons, we extend FOL formulas on graphs to  structures of a special type. We say that a pair $(G,\mathbf{v})$, where $\mathbf{v}=\langle v_1,\ldots,v_r\rangle$ is an $r$-tuple of vertices of $G$, is an \emph{$r$-structure}. Let $\varphi$ be an FOL formula without free variables and let $\mathbf{x}=\langle x_1,\ldots,x_r\rangle$ be an $r$-tuple of  distinct variables  of $\varphi$. We denote by $\varphi[\mathbf{x}]$ the formula obtained from $\varphi$ by the deletion of the quantification over $x_1,\ldots,x_r$, that is, these variables become the free variables of $\varphi[\mathbf{x}]$. For an $r$-structure $(G,\mathbf{v})$ with $\mathbf{v}=\langle v_1,\ldots,v_r\rangle$ and $\varphi[\mathbf{x}]$, we write $(G,\mathbf{v})\models \varphi[\mathbf{x}]$ to denote that $\varphi[\mathbf{x}]$ evaluates to \emph{true} on $G$ if $x_i$ is assigned $v_i$ for $i\in\{1,\ldots,r\}$. If $r=0$, that is, $\mathbf{v}$ and $\mathbf{x}$ are empty, then $(G,\mathbf{v})\models \varphi[\mathbf{x}]$ is equivalent to $G\models \varphi$.

As a subroutine in our algorithms, we have to evaluate  FOL formulas on graph, that is, solve the \probMC problem. Let $\varphi$ be a FOL formula. The task of \probMC is, given a graph $G$, decide whether $G\models \varphi$. It was shown by Vardi~\cite{Vardi82} that 
\probMC is \classPSPACE-complete. The problem is also hard from the parameterized complexity viewpoint when parameterized by the size of the formula. It was proved by Frick and Grohe in~\cite{FrickG04} that the problem is ${\sf AW}[*]$-complete for this parametrization (see, e.g., the book~\cite{FlumG06} for the definition of the class). Moreover, it can be noted that the problem is already \classW{1}-hard for formulas having only existential quantifiers, that is, for $\varphi\in \Sigma_1$, by observing that the existence of an independent set of size $k$ can be easily expressed by such a formula and \textsc{Independent Set} is well-known to be \classW{1}-complete~\cite{DowneyF13}. 
This implies that we cannot expect an FPT algorithm for the problem.  However, it is easy to see that  \probMC is in \classXP when parameterized by the number of variables, because the problem for a formula with $s$ variables can be solved in $\Oh(n^s)$ time by the exhaustive search (the currently best algorithm is given by Williams in~\cite{Williams14}). This explains the exponential dependence of the polynomials in running times in our algorithm on the number of variables.  For referencing, we state the following observation.

\begin{observation}\label{obs:mod-check}
\probMC for an FOL formula $\varphi$ can be solved in $n^{\Oh(|\varphi|)}$ time. 
\end{observation}

\medskip
In monadic second-oder logic (MSOL), we additionally can quantify over sets of vertices and edges. Formally, we can use variables for sets of vertices and edges and have  the predicate $x\in X$, where $x$ is a vertex (an edge, respectively) variable and $X$ a vertex set (an edge set, respectively) variable, denoting that $x$ is an element of $X$. As with FOL formulas, we write $G\models\varphi$ to denote that an MSOL formula $\varphi$ evaluates \emph{true} on $G$. We refer to the book of Courcelle and Engelfriet~\cite{CourcelleE12} for the details of MSOL on graphs.
  
\section{Properties of elimination distance}\label{sec_elim}
In this section we derive the properties of the elimination distances ,  $  \edphi{\edphione}$ and $\edphi{\edphitwo}$ 
that will be used in the proof of the main theorem. We also define $  \edphi{\edphthree}$.

\begin{observation}\label{obs:dist}
For every FOL formula $\varphi$ and every graph $G$, $\edphi{\edphitwo}(G)\leq \edphi{\edphione}(G)+1$. 
\end{observation}

\begin{proof}
The proof is by induction on the value of $\edphi{\edphione}(G)$. 

Suppose that $\edphi{\edphione}(G)=0$ for a graph $G$. If $G$ is connected, then $G\models \varphi$ and $\edphi{\edphitwo}(G)=0$. Hence, the inequality holds.
If $G$ is disconnected, then $C\models\varphi$ for every component $C$ of $G$. If $G\models \varphi$, then $\edphi{\edphitwo}(G)=0$. If $G\not\models\varphi$, then $\edphi{\edphitwo}(G)=\max\{1,\max\{\edphi{\edphitwo}(C)\mid C\text{ is a component of }G\}\}=1$. In both cases,  $\edphi{\edphitwo}(G)\leq \edphi{\edphione}(G)+1$.

Assume that    $\edphi{\edphione}(G)>0$ and $\edphi{\edphitwo}(G')\leq \edphi{\edphione}(G')+1$ for all $G'$ with $\edphi{\edphione}(G')<\edphi{\edphione}(G)$. The claim is trivial if $\edphi{\edphitwo}(G)=0$. Let $\edphi{\edphitwo}(G)>0$. We have two cases.

\subparagraph{Case~1.} $G$ is connected. By definition, there is $u\in V(G)$ such that 
$\edphi{\edphione}(G)=1+\edphi{\edphione}(G-u)$. Because $\edphi{\edphitwo}(G)>0$,
\begin{equation*}
\edphi{\edphitwo}(G)=1+\min_{v\in V(G)}\edphi{\edphitwo}(G-v)\leq 1+\edphi{\edphitwo}(G-u).
\end{equation*}
Then by induction,
\begin{equation*}
\edphi{\edphitwo}(G)\leq 1+\edphi{\edphitwo}(G-u)\leq 2+\edphi{\edphione}(G-u)=1+\edphi{\edphione}(G).
\end{equation*}

\subparagraph{Case~2.} $G$ is disconnected. Let $C_1,\ldots,C_s$ be the components of $G$. By definition, 
$\edphi{\edphione}(G)=\max_{1\leq i\leq s}\edphi{\edphione}(C_i)$. In particular, we have that $\edphi{\edphione}(C_i)\leq \edphi{\edphione}(G)$ for every $i\in\{1,\ldots,s\}$. Notice that by the already proved claim for connected graphs in Case~1, $\edphi{\edphitwo}(C_i)\leq \edphi{\edphione}(C_i)+1$ for every $i\in\{1,\ldots,s\}$.
Because $\edphi{\edphione}(G)>0$, $G\not\models \varphi$. Then 
\begin{align*} 
\edphi{\edphitwo}(G)=&\max\{1,\max_{1\leq i\leq s}\edphi{\edphitwo}(C_i)\}\leq \max\{1,\max_{1\leq i\leq s}(\edphi{\edphione}(C_i)+1)\}\\
=&\max_{1\leq i\leq s}\edphi{\edphione}(C_i)+1=\edphi{\edphione}(G)+1
\end{align*}
as required. This completes the proof. 
\end{proof}

The example of $\varphi=\forall x\forall y~x=y$ and an edgeless graph $G$ with at least two vertices shows that the inequality in Observation~\ref{obs:dist} is tight.
However, $\edphi{\edphione}(G)$ and $\edphi{\edphitwo}(G)$ can be far apart. Consider $\varphi=\exists u\exists v~\neg(u=v)\wedge \neg(u\sim v)$ that defines the property that a graph has two nonadjacent vertices. Let $G$ be the disjoint union of the complete $n$-vertex graph $K_n$ and an isolated vertex. Then $G\models\varphi$ and, therefore, $\edphi{\edphitwo}(G)=0$. From the other side, $G$ is disconnected and it is easy to see that $\edphi{\edphione}(K_n)=n$.
This means that $\edphi{\edphione}(G)=n$ and
$\edphi{\edphione}(G)-\edphi{\edphitwo}(G)=n$, that is, the difference can be arbitrary large. 

For algorithmic purposes, it is convenient for us to define $\edphi{\edphione}(G)$ and $\edphi{\edphitwo}(G)$ via deletions of sets of vertices of $G$ with a special structure. Similar approach was recently exploited by Agrawal et al.~\cite{SridharanPASK21} but we do it in a different way, because we consider two variants of eliminations distances.

Let $G$ be a graph and let $d\geq 0$ be an integer. We say that a set of vertices $X\subseteq V(G)$ is an \emph{elimination set of depth at most $d$} if there is a rooted tree $T$ of depth at most $d$  and a bijective mapping $\alpha\colon V(T)\rightarrow X$ such that 
for every two distinct incomparable nodes $x$ and $y$ of $T$, $\alpha(A_T(v))$ is an $(\alpha(x),\alpha(y))$-separator in $G$, where $v$ is the lowest common ancestor of $x$ and $y$
(recall that $A_T(v)$ denotes the set of ancestors of $v$). 
We also say that the pair $(T,\alpha)$ is a \emph{representation} of $X$  (or   \emph{represents} $X$).  
The \emph{depth} of $X\subseteq V(G)$, denoted $\dpt(X)$, is the minimum $d$ such that $X$ is an elimination set of depth at most $d$.  
 We assume that the empty set is an elimination set of depth $-1$.

We call a representation $(T,\alpha)$ of an elimination set $X\subseteq V(G)$ \emph{nice} if for every nonleaf node $v\in V(T)$ and its child $x$, the vertices of $\alpha(D_T(x))$ are in the same component of $G-A_T(v)$. The following property is useful for us.

\begin{lemma}\label{lem_nice}
Let $G$ be a connected graph and let $d\geq 0$ be an integer. Then a  nonempty $X\subseteq V(G)$ is an elimination set of depth at most $d$ if and only if $X$ has a nice representation $(T,\alpha)$ with $\dpt(T)\leq d$. Moreover, if $(T,\alpha)$ is a representation of $X$, then there is a nice representation $(T',\alpha)$ of $X$ with $V(T')=V(T)$ such that 
(i) $\alpha(L(T))\subseteq \alpha(L(T'))$ and (ii) for each $v\in X$, $\dpt_T(\alpha^{-1}(v))\geq \dpt_{T'}(\alpha^{-1}(v))$.      
\end{lemma}

\begin{proof}
Clearly, if $X\subseteq V(G)$ has a nice representation $(T,\alpha)$ with $\dpt(T)\leq d$, then $X$ is an elimination set of depth at most $d$. For the opposite direction, it is sufficient to show the second claim. Let   $(T,\alpha)$ is a representation of $X$ and $\dpt(T)\leq d$. 
We show the existence of $(T',\alpha)$ satisfying (i) and (ii) by induction on $d$.

The claim is trivial if  $|X|=1$ as $T'=T$ in this case.
  Assume that $|X|\geq 2$ and $d\geq 1$. Denote by $r$ the root of $T$ and let $u=\alpha(r)$. Consider the components $C_1,\ldots,C_s$ of $G-u$ containing at least one vertex of $X$. For every $i\in\{1,\ldots,s\}$, let $X_i=V(C_i)\cap X$ and $U_i=\alpha^{-1}(X_i)$. 
  For every $i\in\{1,\ldots,s\}$, we construct the tree $T_i$ with the set of vertices $U_i\cup\{r\}$ as follows. For every $x\in U_i$ such that $x\neq r$, we find a proper ancestor $y\in U_i$ with respect to $T$ of maximum depth and make $y$ the parent of $x$, and if $x$ has no ancestors in $U_i$, we make $r$ the parent of $x$. Because the choice of the parent is unique, $T_i$ has no cycles, and because we assign the parent to every node distinct  from $r$, we conclude that $T_i$ is a tree. Denote by $\tilde{T}$ the union of $T_1,\ldots,T_s$ and set $r$ be its root. Because every node of $\tilde{T}$ distinct from $r$ got a parent from the set of its proper ancestors in $T$, 
 (i) $\alpha(L(T))\subseteq \alpha(L(\tilde{T}))$ and (ii) for each $v\in X$, $\dpt_T(\alpha^{-1}(v))\geq \dpt_{\tilde{T}}(\alpha^{-1}(v))$.    
  
 We prove that $(\tilde{T},\alpha)$ represents $X$. Consider incomparable nodes $x$ and $y$ of $\tilde{T}$ and denote by $v$ their lowest common ancestor. 
 We have to show that $\alpha(A_{\tilde{T}}(v))$ is an $(\alpha(x),\alpha(y))$-separator in $G$. This is trivial if $\alpha(x)$ and $\alpha(y)$ are in distinct components of $G-u$. Assume that $\alpha(x)$ and $\alpha(y)$ are in the same component $C_i$ for some $i\in\{1,\ldots,s\}$, that is, $x,y\in U_i$. 
 Note that by the construction of $\tilde{T}$, $x$ and $y$ are incomparable in $T$. Let $v'$ be their lowest common ancestor in $T$. Clearly, $v$ is an ancestor of $v$ in $T$.  
 By the construction of $\tilde{T}$, $A_{\tilde{T}}(v)\cap V(C_i)=A_T(v')\cap V(C_i)$. Because $A_T(v)$ separates $\alpha(x)$ and $\alpha(y)$, we have that  $A_T(v)\cap V(C_i)$ is an $(\alpha(x),\alpha(y))$-separator. Therefore,  $\alpha(A_{\tilde{T}}(v))$ is also an $(\alpha(x),\alpha(y))$-separator. This proves that $(\tilde{T},\alpha)$ represents $X$.

Consider $i\in\{1,\ldots,s\}$. Observe that $r$ has a unique child in $U_i$ in $\tilde{T}$. Otherwise, if $x$ and $y$ are distinct children of $r$, we have that $x$ and $y$ have no ancestors in $U_i$ with respect to $T$. Let $v$ be the lowest common ancestor of $x$ and $y$ in $T$. Note that $v\neq x,y$ and $\alpha(A_T(v))$ does not separate $\alpha(x)$ and $\alpha(y)$ contradicting that $(T,\alpha)$ represents $X$. Hence,  $r$ has the  unique child $r_i$ in $U_i$. Let $\tilde{T}_i$ be the subtree of $\tilde{T}$ rooted in $r_i$. We set $\alpha_i(x)=\alpha(x)$ for $x\in U_i$.
Because $(\tilde{T},\alpha)$ represents $X$, it is straightforward to verify that $(\tilde{T}_i,\alpha_i)$ represents $X_j$ in the graph $C_i$. 
Because $\dpt(\tilde{T}_i)\leq d-1$, we can apply the inductive assumption. We obtain that there is a nice representation $(T'_i,\alpha_i)$ of $X_i$ in $C_i$ with $V(T_i')=V(\tilde{T}_i)$ such that 
(i) $\alpha(L(\tilde{T}_i))\subseteq \alpha(L(T_i'))$ and (ii) for each $v\in X_i$, $\dpt_{\tilde{T}_i}(\alpha_i^{-1}(v))\geq \dpt_{T_i'}(\alpha_i^{-1}(v))$.   
Notice that by the second condition, $T_i'$ is rooted in $r_i$.

We construct the trees $T_i'$ for all $i\in\{1,\dots,s\}$ and then construct $T'$ from their union by making $r_1,\ldots,r_s$ the children of $r$.
Clearly, $\dpt(T')\leq d$ and $(T',\alpha)$ is a nice representation of $X$ satisfying conditions (i) and (ii) of the lemma.
\end{proof}

It is also useful to characterize the depths of an elimination set  in a disconnected graph.

\begin{lemma}\label{lem_depth-disconn} 
Let $G$ be a graph with components $C_1,\ldots,C_s$ and let $X\subseteq V(G)$ such that $X\neq\emptyset$. Then
\begin{equation}\label{eq:disc}
\dpt(X)=\min_{1\leq i\leq s}\max\{\dpt(X_i),\max\{\dpt(X_j)\mid 1\leq j\leq s,~j\neq i\}+1 \},
\end{equation}
where $X_i=X\cap V(C_i)$ for $i\in\{1,\ldots,s\}$.
\end{lemma}

\begin{proof}
Recall that $\dpt(\emptyset)=-1$ by definition. This allows us to assume without loss of generality that $X_i\neq\emptyset$ for all $i\in\{1,\ldots,s\}$. Otherwise, we can delete each component $C_i$ such that $X_i=\emptyset$ without violating the value of $\dpt(X)$ and the right part of (\ref{eq:disc}).

To show that $$\dpt(X)\leq\min_{1\leq i\leq s}\max\{\dpt(X_i),\max\{\dpt(X_j)\mid 1\leq j\leq s,~j\neq i\}+1 \},$$ assume that the minimum value of the right part of (\ref{eq:disc}) is achieved for $i\in\{1,\ldots,s\}$. For every  $j\in\{1,\ldots,s\}$, let $(T_j,\alpha_j)$ be a representation of $X_j$ in $C_j$,
where $T_j$ is rooted in $r_j$ and  $\dpt(T_j)=\dpt(X_j)$. We construct the tree $T$ with the root $r=r_i$ from $T_1,\ldots,T_s$ by making each $r_j$ 
for $j\in\{1,\ldots,s\}\setminus\{i\}$ a child of $r$. Clearly,
$\dpt(T)=\max\{\dpt(T_i),\max\{\dpt(T_j)\mid 1\leq j\leq s,~j\neq i\}+1 \}=\max\{\dpt(X_i),\max\{\dpt(X_j)\mid 1\leq j\leq s,~j\neq i\}+1 \}$. 
We define $\alpha\colon V(T)\rightarrow X$ by setting
$\alpha(x)=\alpha_i(x)$ whenever $x\in X_i$ for some $i\in\{1,\ldots,s\}$.
It is straightforward to verify that $(T,\alpha)$ represents $X$.

To show the opposite inequality  $$\dpt(X)\geq\min_{1\leq i\leq s}\max\{\dpt(X_i),\max\{\dpt(X_j)\mid 1\leq j\leq s,~j\neq i\}+1 \},$$ let $(T,\alpha)$ be a representation of $X$, where $T$ is rooted in $r$ and $\dpt(T)=\dpt(X)$. By symmetry, we assume without loss of generality that $r\in V(C_1)$. 
For every $j\in\{1,\ldots,s\}$, let $U_i=\alpha^{-1}(X_j)$. 

The rest of the proof is done similarly to the proof of Lemma~\ref{lem_nice}. 
For every $j\in\{1,\ldots,s\}$, we construct the tree $T_j$ with the set of vertices $U_j\cup\{r\}$ as follows. For every $x\in U_j$ such that $x\neq r$, we find a proper ancestor $y\in U_j$ of $x$ with respect to $T$ of maximum depth and make $y$ the parent of $x$, and if $x$ has no ancestors in $U_j$, we make $r$ the parent of $x$. Because the choice of the parent is unique, $T_j$ has no cycles, and because we assign the parent to every node distinct  from $r$, we conclude that $T_j$ is a tree. Denote by $T'$ the union of $T_1,\ldots,T_s$ and set $r$ be the root. Because every node of $T'$ distinct from $r$ got a parent from the set of its proper ancestors in $T$, $\dpt(T')\leq\dpt(T)=\dpt(X)$. 

We claim that $(T',\alpha)$ represents $X$. To show this, let $x$ and $y$ be incomparable nodes of $T'$ and let $v$ be their lowest common ancestor. 
We  show that $\alpha(A_{T'}(v))$ is an $(\alpha(x),\alpha(y))$-separator in $G$. This is trivial if $\alpha(x)$ and $\alpha(y)$ are in distinct components of $G$. Assume that $\alpha(x)$ and $\alpha(y)$ are in the same component $C_j$ for some $j\in\{1,\ldots,s\}$, that is, $x,y\in U_j$. Notice that by the construction of $T'$, $A_{T'}(v)\cap C_j=A_T(v')\cap C_j$, where $v'$ is the lowest common ancestor of $x$ and $y$ in $T$. Because $A_T(v')$ separates $\alpha(x)$ and $\alpha(y)$, we have that  $A_T(v')\cap C_j$ is an $(\alpha(x),\alpha(y))$-separator. Therefore,  $\alpha(A_{T'}(v))$ is an $(\alpha(x),\alpha(y))$-separator as well, as required.

 Let $j\in\{2,\ldots,s\}$. Observe that $r$ has a unique child in $U_j$ in $T'$. Otherwise, if $x$ and $y$ are distinct children of $r$, we have that $x$ and $y$ have no ancestors in $U_j$. Let $v$ be the lowest common ancestor of $x$ and $y$ in $T$. Note that $v\neq x,y$ and $\alpha(A_T(v))$ does not separate $\alpha(x)$ and $\alpha(y)$ contradicting that $(T,\alpha)$ represents $X$. Hence,  $r$ has the  unique child $r_j$ in $U_i$. Let $T_j'$ be the subtree of $T'$ rooted in $r_j$. Define $\alpha_j(x)=\alpha(x)$ for $x\in U_j$.
 Since $(T',\alpha)$ represents $X$, we obtain that $(T_j',\alpha_j)$ represents $X_j$. Then 
 $\dpt(X_j)\leq \dpt(T_j')\leq \dpt(T)-1=\dpt(X)-1$. 
 
It is straightforward to verify that $(T_1,\alpha_1)$ represents $X_1$, where $\alpha_1(x)=\alpha(x)$ for $x\in U_1$. This means that 
$\dpt(X_1)\leq \dpt(T_1)\leq \dpt(X)$.  Because   $\dpt(X_j)+1\leq \dpt(X)$ for $j\in\{2,\ldots,s\}$,  
\begin{align*}
\dpt(X)\geq&\max\{\dpt(X_1),\dpt(X_2)+1,\ldots,\dpt(X_s)+1\}\\
\geq &\min_{1\leq i\leq s}\max\{\dpt(X_i),\max\{\dpt(X_j)\mid 1\leq j\leq s,~j\neq i\}+1 \}.
\end{align*}
This completes the proof.
\end{proof}

 It is sufficient for our purposes to characterize $\edphi{\edphione}(G)$ for connected graphs and we do it in the following lemma.

\begin{lemma}\label{lem_elim-dist-one} 
Let $\varphi$ be an FOL formula and let $G$ be a connected graph. 
Let also $d\geq 0$ be an  integer. Then $\edphi{\edphione}(G)\leq d$ if and only if $G$ contains an elimination set $X$ of depth at most $d-1$ such that $C\models \varphi$ for every component $C$ of $G-X$.
\end{lemma} 

\begin{proof}
First, we show that if $\edphi{\edphione}(G)\leq d$, then $G$ has an elimination set $X$ of depth at most $d-1$ such that $C\models \varphi$ for every component $C$ of $G-X$. The proof is by induction on $d$.
The claim is trivial if $\edphi{\edphione}(G)=0$ as $\dpt(\emptyset)=-1$ by the definition. Let $d\geq \edphi{\edphione}(G)\geq 1$.

Because $G$ is connected and $\edphi{\edphione}(G)>0$, there is $v\in V(G)$ such that $\edphi{\edphione}(G)=1+\edphi{\edphione}(G-v)$. We construct a node $r$ of $T$ and set it be the root. If $C\models\varphi$ for every component $C$ of $G-v$,  then the construction of $X$ and $T$ is completed and we define $\alpha(r)=v$. Otherwise, let $C_1,\ldots,C_s$ be the components of $G-v$ such that $C_i\not\models\varphi$ for $i\in\{1,\ldots,s\}$. 
Clearly, $\edphi{\edphione}(C_i)\leq d-1$ for $i\in\{1,\ldots,s\}$.  
Let $i\in\{1,\ldots,s\}$. By induction, there is  an elimination set $X_i\subseteq V(C_i)$ of depth  at most $d-2$ such that $H\models \varphi$ for every component $H$ of $C_i-X_i$. Then there is a corresponding 
representation $(T_i,\alpha_i)$ of $X_i$ in $C_i$. Let $r_i$ be the root of $T_i$.
We define $X=\{v\}\cup\bigcup_{i=1}^sX_i$, and construct $T$ from $T_1,\ldots,T_s$ by making $r_1,\ldots,r_s$ the children of $r$. Finally, we define 
\begin{equation*}  
\alpha(x)=
\begin{cases}
v,& \mbox{if }x=r,\\
\alpha_i(x),&\mbox{if }x\in V(C_i)\text{ for some }i\in\{1,\ldots,s\}.
\end{cases}
\end{equation*}
It is straightforward to verify that $X$ is an elimination set of depth at most $d-1$ with respect to $(T,\alpha)$.

For the opposite direction, we assume that $X$ is an elimination set of minimum depth  such that $C\models \varphi$ for every component $C$ of $G-X$.
We assume that  the depth of $X$ is $d-1$ 
 and prove that $\edphi{\edphione}{G}\leq d$. The proof is by induction on $d$.

The claim is trivial if $d=0$, that is, if $X=\emptyset$. 
Suppose that $d=1$, that is, the depth of an elimination set $X$ is zero and, therefore, $X=\{u\}$ for some $u\in V(G)$. We have that $C\models\varphi$ for every component $C$ of $G-u$. This means that $\edphi{\edphione}(C)=0$ for every component $C$ and, therefore, $\edphi{\edphione}(G-u)=0$.
Then because $\edphi{\edphione}(G)>0$, 
$\edphi{\edphione}(G)=1+\min_{v\in V(G)}\edphi{\edphione}(G-v)=1+\edphi{\edphione}(G-u)=1\leq d$. 

Suppose that $d\geq 2$ and the claim holds for the lesser values of $d$. 
Because $G$ is connected, by Lemma~\ref{lem_nice}, there is a nice represenation $(T,\alpha)$ of $X$ with $\dpt(T)=d-1$. 
Let $r$ be the root of $T$ and $u=\alpha(r)$. Because $\edphi{\edphione}(G)>0$, 
$\edphi{\edphione}(G)=1+\min_{v\in V(G)}\edphi{\edphione}(G-v)\leq 1+\edphi{\edphione}(G-u)$ and it is sufficient to show that $\edphi{\edphione}(G-u)\leq d-1$. 
For this, we have to prove that  $\edphi{\edphione}(C)\leq d-2$ for every component $C$ of 
 $G-u$.  
 
 If $V(C)\cap X=\emptyset$ for a component $C$, then $C$ is a component of $G-X$ and we have that $C\models \varphi$. Then  
$\edphi{\edphione}(C)=0\leq d-2$. Consider the components $C_1,\ldots,C_s$ of $G-u$ such that $V(C_i)\cap X\neq\emptyset$.  
Because $(T,\alpha)$ is nice, $r$ has $s$ children $x_1,\ldots,x_s$ such that for every $i\in\{1,\ldots,s\}$, $\alpha(V(T_i))\subseteq V(C_i)$, where $T_i$ is the subtree of $T$ rooted in $x_i$. Let $\alpha_i\colon V(T_i)\rightarrow V(C_i)$ be the restriction of $\alpha$ on $V(T_i)$ for $i\in\{1,\ldots,s\}$. Consider $i\in\{1,\ldots,s\}$. We have that $(T_i,\alpha_i)$ is a representation of $X_i=X\cap V(C_i)$. Notice that for each component $C$ of $C_i-X_i$, $C\models \varphi$. Clearly, $\dpt(T_i)< \dpt(T)$. This implies that we can use the inductive assumption  and conclude that $\edphi{\edphione}(C_i)\leq d-1$. Therefore, $\edphi{\edphione}(G-u)\leq d-1$ and this concludes the proof.
\end{proof}

To characterize $\edphi{\edphitwo}(G)$, we need additional definitions. 

Let $G$ be a connected graph and let $X\subseteq V(G)$ be an elimination set represented by $(T,\alpha)$. We say that a node $x\in V(T)$ is an \emph{anchor} of a component $C$ of $G-X$ if $x$ is the node  of maximum depth in $T$ such that $\alpha(x)\in N_G(V(C))$. We also say that $C$ is \emph{anchored} in $x$. Notice that the definition of an elimination set immediately implies the following property.

\begin{observation}\label{obs:anchor} 
Let $G$ be a connected graph and let $X\subseteq V(G)$ be an elimination set represented by $(T,\alpha)$. Then for every component $C$ of $G-X$, 
$N_G(V(C))\subseteq \alpha(A_T(x))$, where $x$ is an anchor of $C$.
\end{observation}

In particular, Observation~\ref{obs:anchor} implies that an anchor of each component of $G-X$ is unique. For a node $x\in V(T)$, we denote by $\property_x$ the set of components of $G-X$ anchored in $x$, and $G_x$ denotes the subgraph of $G$ induced by the vertices of the graphs of $\property_x$, that is, $G_x$ is the union of the components of $G-X$ anchored in $x$. Clearly, $\property_x$ and $G_x$ may be empty. Note that the anchors of the components of $G-X$ depend on the choice of a representation. Therefore, we use the above notation only when $(T,\alpha)$ is fixed and clear from the context.

 \begin{lemma}\label{lem_elim-dist-two} 
Let $\varphi$ be an FOL formula and let $G$ be a connected graph with $\edphi{\edphitwo}(G)>0$. Let also $d$ be a positive integer. Then $\edphi{\edphitwo}(G)\leq d$ if and only if $G$ contains an elimination set $X$ of depth at most $d-1$ with a representation $(T,\alpha)$ such that the following is fulfilled: 
\begin{itemize}
\item[(i)] for every nonleaf node $x\in V(T)$, $C\models \varphi$ for every  $C\in\property_x$,
\item[(ii)] for every leaf $x$ of $T$ with $\dpt_T(x)\leq d-2$, either $G_x\models \varphi$ or $C\models \varphi$ for every  $C\in\property_x$,
\item[(iii)] for every leaf $x$ of $T$ with $\dpt_T(x)=d-1$,  $G_x\models \varphi$.
\end{itemize}
\end{lemma} 

\begin{proof}
The lemma is proved similarly Lemma~\ref{lem_elim-dist-one}. We begin with showing that if $\edphi{\edphitwo}(G)\leq d$, then $G$ has an elimination set $X$ of depth at most $d-1$ with a representation $(T,\alpha)$ such that conditions (i)--(iii) are fulfilled. For this, we inductively construct $X$ and $(T,\alpha)$ with $\dpt(T)\leq d-1$ using the definition of $\edphi{\edphitwo}(G)$. 

Since $G$ is connected and $\edphi{\edphitwo}(G)>0$, there is $v\in V(G)$ such that $\edphi{\edphitwo}(G)=1+\edphi{\edphione}(G-v)$. We construct a node $r$ of $T$ and set it be the root. If either $G-v\models\varphi$
or $C\models\varphi$ for every component $C$ of $G-v$,  then the construction of $X$ and $T$ is completed and we define $\alpha(r)=v$. 
Note that $\edphi{\edphitwo}(G)=1$  in the first case  and $\edphi{\edphitwo}(G)=2$  in the second. This implies that (i)--(iii) are fulfilled. Assume from now that this is not the case.

Denote by  $C_1,\ldots,C_s$  the components of $G-v$ such that $C_i\not\models\varphi$ for $i\in\{1,\ldots,s\}$. By definition,  $\edphi{\edphitwo}(C_i)\leq d-1$ for $i\in\{1,\ldots,s\}$.  Notice that for each component $C$ of $G-v$ distinct from $C_1,\ldots,C_s$, $C\models\varphi$.
Then we can assume inductively that for every $i\in\{1,\ldots,s\}$, there is  an elimination set $X_i\subseteq V(C_i)$ of depth  at most $d-2$ with respect to  $C_i$  with a representation  $(T_i,\alpha_i)$ such that conditions (i)--(iii) are fulfilled. Let $r_i$ be the root of $T_i$ for $i\in\{1,\ldots,s\}$. 
We define $X=\{v\}\cup\bigcup_{i=1}^sX_i$, and construct $T$ from $T_1,\ldots,T_s$ by making $r_1,\ldots,r_s$ the children of $r$. Then we set 
\begin{equation*}  
\alpha(x)=
\begin{cases}
v,& \mbox{if }x=r,\\
\alpha_i(x),&\mbox{if }x\in V(C_i)\text{ for some }i\in\{1,\ldots,s\}.
\end{cases}
\end{equation*}
Using the inductive assumptions that (i)--(iii) are fulfilled for  $X_i$ with $(T_i,\alpha_i)$ for every $i\in\{1,\ldots,s\}$ and the observation that $\property_r$ consists of the components of $G-v$ distinct from $C_1,\ldots,C_s$, we obtain that (i)--(iii) are fulfilled for $X$ and the representation $(T,\alpha)$.  

To show the implication in the opposite direction,  assume that $X$ is an elimination set of depth at most $d-1$ with a representation $(T,\alpha)$ satisfying (i)--(iii). By the second claim of Lemma~\ref{lem_nice}, we can assume that $T$ is nice. We show that $\edphi{\edphitwo}(G)\leq d$ by the induction on $\dpt(T)$.

Suppose that $\dpt(T)=0$, that is, the depth of an elimination set $X$ is zero and, therefore, $X=\{u\}$ for some $u\in V(G)$.
If $d=1$, then  $G_u\models \varphi$ and $\edphi{\edphitwo}(G)=1$. If $d\geq 2$, then either  $G_u\models \varphi$ or $C\models \varphi$ for every component $C$ of $G-u$. In both cases, $\edphi{\edphitwo}(G)\leq 2$ by the definition of $\edphi{\edphitwo}(G)$.

Assume that $\dpt(T)\geq 1$. In particular, $d\geq 2$. Since $G$ is connected and $\edphi{\edphitwo}(G)>0$, 
$\edphi{\edphitwo}(G)=1+\min_{v\in V(G)}\edphi{\edphitwo}(G-v)\leq 1+\edphi{\edphitwo}(G-u)$ and it is sufficient to show that $\edphi{\edphitwo}(G-u)\leq d-1$. 
If $G-u\models\varphi$, then $\edphi{\edphitwo}(G)=1\leq d$. Assume from now that $G-u\not\models\varphi$. Then, by the definition of $\edphi{\edphione}(G)$, it is sufficient to show  that  $\edphi{\edphitwo}(C)\leq d-2$ for every component $C$ of $G-u$.

If $V(C)\cap X=\emptyset$ for a component $C$ of $G-u$, then $C\in \property_r$ and  $C\models \varphi$. Then  
$\edphi{\edphitwo}(C)=0\leq d-2$. Consider the components $C_1,\ldots,C_s$ of $G-u$ such that $V(C_i)\cap X\neq\emptyset$.  
Because $(T,\alpha)$ is nice, $r$ has $s$ children $x_1,\ldots,x_s$ such that for every $i\in\{1,\ldots,s\}$, $\alpha(V(T_i))\subseteq V(C_i)$, where $T_i$ is the subtree of $T$ rooted in $x_i$. Let $\alpha_i\colon V(T_i)\rightarrow V(C_i)$ be the restriction of $\alpha$ on $V(T_i)$ for $i\in\{1,\ldots,s\}$. Consider $i\in\{1,\ldots,s\}$. We have that $(T_i,\alpha_i)$ is a representation of $X_i=X\cap V(C_i)$ satisfying (i)--(iii). 
Notice that $\dpt(T_i)< \dpt(T)$. Then by the inductive assumption  $\edphi{\edphitwo}(C_i)\leq d-1$. Therefore, $\edphi{\edphitwo}(G-u)\leq d-1$ and this concludes the proof.
\end{proof}
 
 Lemmas~\ref{lem_elim-dist-one} and~\ref{lem_elim-dist-two} demonstrate that $\edphi{\edphione}(G)$ and $\edphi{\edphitwo}(G)$, respectively, can be defined via the deletion of  an elimination set. We also use these results in order to define a third variant of the elimination distance.
 
 \begin{definition}[Elimination distance  $\edphi{\edphthree}$]
  Let $\varphi$ be a FOL formula.
 For a graph $G$, $\edphi{\edphthree}(G)$ is the minimum $d$ such that $G$ has an elimination set $X\subseteq V(G)$ of depth $d-1$ such that $G-X\models\varphi$. 
 \end{definition}
 Notice that if the trees in the considered representations  of elimination sets are constrained  to be paths, then we obtain the classical \emph{deletion distance}, that is, the minimum size of a set $X\subseteq V(G)$ such that $G-X\models \varphi$. 
 
Given a FOL formula $\varphi$,  we define the following three variants of the \textsc{Elimination Distance} problems for $\star\in\{\edphione,\edphitwo,\edphthree\}$:

\defparproblema{\probED{\star}}{A graph $G$ and a nonnegative integer $k$.}{$k$}{Decide whether $\edphi{\star}(G)\leq k$.} 
 
 These problems may be seen as generalizations of \probDEL problem for a formula $\varphi$ that asks, given a graph $G$ and a nonnegative integer $k$, whether there is a set $S$ of size at most $k$ such that $G-S\models\varphi$. In particular, Observation~\ref{obs:anchor} implies the following.
 
 \begin{observation}\label{obs:deletion}
 \probDEL and \probED{\star} for $\star\in\{\edphione,\edphitwo,\edphthree\}$ are equivalent on instances $(G,k)$, where $G$ is a $(k+1)$-connected graph. 
 \end{observation}

\section{An FPT algorithm for $\Sigma_3$-formulas}\label{sec_FPT}
In this section, we show the main algorithmic result, Theorem~\ref{thm_algorithm_informal},  that \probED{\star} is \classFPT for formulas from $\Sigma_3$. Now we state this theorem formally. 

\medskip\noindent\textbf{Theorem~1.}\emph{
For every FOL formula $\varphi\in\Sigma_3$, $\probED{\star}$ can be solved in $f(k)\cdot n^{\Oh(|\varphi|)}$ time for each $\star\in\{\edphione,\edphitwo,\edphthree\}$. }

 \medskip
 
We prove the theorem using the \emph{recursive understanding} technique introduced by Chitnis et al.~\cite{ChitnisCHPP16}. It was recently demonstrated by Agrawal et al.~\cite{SridharanPASK21} that this approach is useful for elimination problems. 
 As we are interested in the quality result, we apply the meta theorem of Lokshtanov et al.~\cite{LokshtanovR0Z18} (see the arxiv version~\cite{LokshtanovR0Z18a} for more details). This simplifies the arguments, but makes the proof nonconstructive. Moreover, we only show the existence of nonuniform FPT algorithms. 
However, it is possible to show the theorem in constructive way by giving uniform algorithm by either using the original approach of Chitnis et al.~\cite{ChitnisCHPP16} or the dynamic programming scheme proposed by Cygan et al.~\cite{CyganLPPS19}.

The remaining part of the section contains the proof of Theorem~\ref{thm_algorithm_informal}. In Subsection~\ref{sec_recursive}, we introduce the notation and provide auxiliary results needed to apply the recursive understanding technique, and in Subsection~\ref{sec_unbreak}, we prove that the \probED{\star} is \classFPT for the key case when the input graphs cannot be partitioned in big parts by separators of bounded size.   
 
\subsection{Recursive understanding}\label{sec_recursive}
Let $G$ be a graph. A pair $(A,B)$, where $A,B\subseteq V(G)$ and $A\cup B=V(G)$, is called a \emph{separtion} of $G$ if there is no edge $uv$ with $u\in A\setminus B$ and $v\in B\setminus A$.   In other words, $A\cap B$ is a $(u,v)$-separator for every $u\in A\setminus B$ and $v\in B\setminus A$.
The \emph{order} of $(A,B)$ is $|A\cap B|$.

Let $p,q$ be positive integers. A graph $G$ is said to be \emph{$(p,q)$-unbreakable} if for every separation $(A,B)$ of $G$ of order at most $q$, either $|A\setminus B|\leq p$ or $|B\setminus A|\leq p$, that is, $G$ has no separator of size at most $q$ that partitions the graph into two parts of size at least $p+1$ each.

We state the restricted variant of the meta theorem of  Lokshtanov et al.~\cite{LokshtanovR0Z18}. Lokshtanov et al. proved the theorem  for structures and counting monadic second-order logic.  For us, it is sufficient to state the theorem for graphs and MSOL.

\begin{theorem}[{\cite[Theorem~1]{LokshtanovR0Z18}} ]\label{thm_meta}
Let $\psi$ be a MSOL formula. For all $q\in \mathbb{N}$, there exists $p\in\mathbb{N}$ such that if there exists an algorithm that solves  \probMC for $\psi$ on $(p,q)$-unbreakable graphs in $\Oh(n^d)$ time for some $d\geq 4$, then \probMC for $\psi$ can be solved on general graphs in $\Oh(n^d)$ time.  
\end{theorem}

It is crucial that the considered problems may be expressed in MSOL.

\begin{lemma}\label{lem_expr}
For every FOL formula $\varphi$, every $\star\in\{\edphione,\edphitwo,\edphthree\}$, and every integer $k\geq 0$, there is a MSOL formula $\psi_k^{\star}$ such that for each graph $G$, $G\models \psi_k^{\star}$ if and only if $\edphi{\star}(G)\leq k$. 
\end{lemma}

\begin{proof}
We use capital letter to write vertex set variables and small letters are used for vertex variables. To simplify notation, we introduce some auxiliary formulas. 
Notice that we can express that $Z=X\cap Y$ in MSOL and we write  $X\cap Y$ for such an expression. Similarly, we write $X-Y$ to express that $Z=X\setminus Y$, and we write $X-y$ for $X\setminus\{y\}$. Also $\overline{X}$ is used for the complement of $X$.
It is well-known that the connectivity property can be expressed in MSOL, because of the following observation: a set $X\subseteq V(G)$ induces a connected subgraph of $G$ if and only if for every partition $(U,W)$ of $X$, there is an edge $uw\in E(G)$ such that $u\in U$ and $w\in W$. Then we can observe that for every $X\subseteq V(G)$, $G[X]$ is a component of $G$ if and  only if $X$ induces a connected subgraph but for every $v\in V(G)\setminus X$, $G[X\cup\{v\}]$ is not a connected graph. 
This 
allows us to use the MSOL formula $\comp(X)$ with a free variable $X$ expressing the property that $X$ induces  a component.  Clearly, every FOL formula is a MSOL formula. In particular, this means that we can construct the MSOL formula $\varphi(X)$  for a free variable $X$ expressing the property that the subgraph induced by $X$ models $\varphi$.  

First, we show the lemma for $\star\in\{\edphione,\edphitwo\}$ using the definitions. 
For this, we inductively  construct  $\psi_k^{\edphione}$ and $\psi_k^{\edphitwo}$. 

It is easy to see that for $k=0$, $\psi_0^{\edphione}=\forall X~\comp(X)\rightarrow \varphi(X)$. 

Now let $k\geq 1$ and assume that $\psi_{k-1}^{\edphione}$ is constructed. Then we can define the MSOL formula $\psi_{k-1}^{\edphione}(X)$ for 
a free variable $X$ expressing the property that the subgraph induced by $X$ models $\psi_{k-1}^{\edphione}$. Then it is straightforward to verify that 
\begin{equation*}
\psi_k^{\edphione}=\psi_{k-1}^{\edphione}\vee(\forall X~\comp(X)\rightarrow (\exists x (x\in X)\wedge \psi_{k-1}^{\edphione}(X-x))).
\end{equation*}

Next, we construct $\psi_k^{\edphitwo}$ for $k\geq 0$.
It is straightforward to see that $\psi_0^{\edphitwo}=\varphi$ and
\begin{equation*}
\psi_1^{\edphitwo}=\psi_0^{\edphitwo}\vee\big(\forall X~\comp(X)\rightarrow (\psi_{0}^{\edphitwo}(X)\vee(\exists x~(x\in X) \wedge \psi_{0}^{\edphitwo}(X-x)))\big).
\end{equation*}
Then for $k\geq 2$,
\begin{equation*}
\psi_k^{\edphitwo}=\psi_{k-1}^{\edphitwo}\vee\big(\forall X~\comp(X)\rightarrow (\exists x~ (x\in X) \wedge \psi_{k-1}^{\edphitwo}(X-x))
\big),
\end{equation*}
where $\psi_{k-1}^{\edphitwo}(X)$ for 
a free variable $X$ expresses the property that the subgraph induced by $X$ models $\psi_{k-1}^{\edphitwo}$.

Finally, we prove the claim for $\psi_k^{\edphthree}$. Here, the proof is more complicated and uses Lemmas~\ref{lem_nice} and \ref{lem_depth-disconn}. We express the property that $X$ is an elimination set of set at most $d$. 

By Lemma~\ref{lem_nice}, if $G$ is a connected graph and $d\geq 0$, then $\dpt(X)\leq d$ if and only if $X$ has a nice representation of depth at most $d$. For a free variable $X$ and an integer $d\geq -1$, we define the formula $\xi_d(X)$ expressing that $X$ has a nice representation $(T,\alpha)$ of depth at most $d$. 
For $d=-1$, $\xi_d(X)=(X=\emptyset)$, and for $d=0$, $\xi_d(X)=(|X|=1)$ by the definition (clearly, the property $|X|=1$ can be expressed in MSOL). Assume that $d\geq 1$ and $\xi_{d-1}(X)$ is already constructed. Additionally, we assume that we are given the formula $\xi_{d-1}(X,Y)$ which expresses the property that $X$ 
has a nice representation  of depth at most $d-1$ in the subgraph induced by $Y$. For this, we observe that  $\xi_{d-1}(X,Y)$ can be constructed from  $\xi_{d-1}(X)$ in a straightforward way. Also we use $\comp(Y,x)$ to denote the formula expressing that $Y$ induces a component of the subgraph obtained by the deletion of $x$.
Then
\begin{equation*}
\xi_d(X)=\xi_{d-1}(X)\vee\big(\exists x~(x\in X)\wedge(\forall Y~(\comp(Y,x)\wedge (X\cap Y\neq\emptyset))\rightarrow \xi_{d-1}(X\cap Y,Y))
\big).
\end{equation*}
To see this, it is sufficient to observe that $\exists x~(x\in X)\wedge(\forall Y~(\comp(Y,x)\wedge (X\cap Y\neq\emptyset))\rightarrow \xi_{d-1}(X\cap Y,Y))$ expresses that $G$ has a vertex $x\in X$ such that for the root $r$ of $T$, $x=\alpha(r)$, and in each component $C$ of $G-x$ containing some vertices of $X$, there is a subtree of $T$ of depth at most $d-1$ that can be used to represent $V(C)\cap X$ in $C$.  

Now we construct the formula $\tilde{\xi}_d$ which expresses that $X$ is an elimination set of depth at most $d$ using Lemma~\ref{lem_depth-disconn}. It is easy to see that $\tilde{\xi}_{-1}(X)=(X=\emptyset)$ and $\tilde{\xi}_d(X)=(|X|=1)$. Assume that $d\geq 1$, $\tilde{\xi}_{d-1}(X)$ is already constructed, and we have a formula $\tilde{\xi}_{d-1}(X,Y)$ expressing that  the depth of $X$ is at most $d-1$
 in the subgraph induced by $Y$. Then by Lemma~\ref{lem_depth-disconn}, 
 \begin{align*}
\tilde{\xi}_d(X)=&\tilde{\xi}_{d-1}(X)\\
\vee& (\exists Y~\comp(Y)\wedge \tilde{\xi}_{d}(X\cap Y,Y)\wedge (\forall Z~(\comp(Z)\wedge (Z\neq Y))\rightarrow \tilde{\xi}_{d-1}(X\cap Z,Z) ) ).
\end{align*}

Using $\tilde{\xi}_d$ for $d\geq -1$, we can write $\psi_k^{\edphthree}$ for $k\geq 0$ as follows
\begin{equation*} 
\psi_k^{\edphthree}=\exists X~\tilde{\xi}_{k-1}\wedge{\varphi}(\overline{X}). 
\end{equation*}
This completes the proof. 
\end{proof}

Theorem~\ref{thm_meta} and Lemma~\ref{lem_expr} allows us to reduce the proof of Theorem~\ref{thm_algorithm_informal} to  solving \probED{\star} for $\star\in\{\edphione,\edphitwo,\edphthree\}$ on unbreakable graphs. For this, we show that any elimination set in an unbreakable graph has bounded size.

\begin{lemma}\label{lem_es_bound}
Let $G$ be a $(p,q)$-unbreakable graph for positive integers $p$ and $q$ with $|V(G)|> (3p+2q)(p+1)$. Let also $X\subseteq V(G)$ 
 be an elimination set of depth at most $d\leq q-1$. Then $|X|\leq p+q$. Furthermore,  there is a unique component $C$ of $G-X$ with at least $p+1$ vertices and $|V(G)\setminus N_G[V(C)]|\leq p$.   
\end{lemma}

\begin{proof}
Let $(T,\alpha)$ be a representation of $X$ with $\dpt(T)\leq d$. Denote by $r$ the root of $T$.

First, we show the weaker bound $|X|\leq 3p+2q$. 

For the sake of contradiction, assume that $|X|\geq 3p+2q+1$. Because $T$ is a tree, it has a node $x$ such that every component of $T-x$ has at most $\frac{1}{2}|V(T|$ nodes.  Let $S=A_T(x)$ and $S'=\alpha(A_T(x))$. Since $\dpt(T)\leq d \leq q-1$, $|S|=|S'|\leq q$.
By the definition of a representation, for every two distinct components  $C$ and $C'$ of $T-S$, and every $x\in \alpha(V(C))$ and $y\in\alpha(V(C'))$, $S'$ is an $(x,y)$-separator in  $G$.     

We now claim that every component $C$ of $T-S$ has at most $p$ nodes.
Suppose to the contrary that there is a component $C$ of $T-S$ with at least $p+1$ nodes. Consider the components $C_1,\ldots,C_s$ of $G-S'$ such that 
$V(C_i)\cap \alpha(V(C))\neq \emptyset $.
Define $A=S\cup\bigcup_{i=1}^sV(C_i)$. Note that $|A\setminus S|\geq |V(C)|\geq p+1$. 
Let $Y=V(T)\setminus (S\cup V(C))$.
 By the choice of $x$, $|V(C)|\leq \frac{1}{2}|V(T)|$. 
Then $|Y|\geq \frac{1}{2}|V(T)|-q\geq (\frac{3}{2}p+q+\frac{1}{2})-q=\frac{3}{2}p+\frac{1}{2}\geq p+1$.  
Observe that for every node $y\in Y$,  $\alpha(y)\notin V(C_i)$ for $i\in\{1,\ldots,s\}$. Then $\alpha(Y)\subseteq V(G)\setminus A$ and $|V(G)\setminus A|\geq p+1$. For $B=(V(G)\setminus A)\cup S'$, we have that $(A,B)$ is a separation of $G$ with $S'=A\cap B$. In particular, $(A,B)$ is a separation of order at most $q$. However, $|A\setminus B|\geq p+1$ and $|B\setminus A|$. This contradicts the unbreakability condition and the claim follows. 

Denote by $C_1,\ldots,C_s$ the components of $T-S$.  Consider a set of indices $I\subseteq\{1,\ldots,s\}$ such that $|\bigcup_{i\in I}V(C_i)|\geq p+1$ and for every proper $I'\subset I$,
$|\bigcup_{i\in I'}V(C_i)|\leq p$. Such a set $I$ exists, because 
$|\bigcup_{i=1}^sV(C_i)|\geq |V(T)|-q\geq 3p+q+1$.
Since each component has at most $p$ nodes, we have that $|\bigcup_{i\in I}V(C_i)|\leq 2p$. Then,
because $|V(T)\setminus S|\geq 3p+1$, $|\bigcup_{i\in \{1,\ldots,s\}\setminus I}V(C_i)|\geq p+1$. 

Consider the components $C_1',\ldots,C_t'$ of $G-S'$ that contain at least one  vertex of $\alpha(V(C_i))$ for some $i\in I$. Define 
$A=S'\cup\bigcup_{i=1}^tV(C_i')$. Note that $|A\setminus S'|\geq p+1$, because  $|\bigcup_{i\in I}V(C_i)|\geq p+1$.
Let $B=(V(G)\setminus A)\cup S'$. Since $\bigcup_{i\in \{1,\ldots,s\}\setminus I}\alpha(V(C_i))\subseteq B\setminus S$, $|B\setminus S|\geq p+1$. 
Then we obtain that $(A,B)$ is a separation of $G$ of order at most $q$ with $|
A\setminus B|\geq p+1$ and $|B\setminus A|\geq p+1$; a contradiction. This concludes the proof of our claim that $|X|\leq 3p+2q$. 

Now we improve the obtained upper bound. Because $|X|\leq 3p+2q$ and $|V(G)|> (3p+2q)(p+1)$, $|V(G)\setminus X|>(3p+2q)p$. Observe that for the set of leaves $L(T)$, we have that 
$|L(T)|\leq 3p+2q$.  By Observation~\ref{obs:anchor}, it holds that for every component $C$ of $G-X$, $N_G(V(C))\subseteq A_T(x)$ for some $x\in L(T)$. By the pigeon hole principle, we conclude that there is $x\in L(T)$ such that for the components $C_1,\ldots,C_s$ of $G-X$ with 
$N_G(V(C_i))\subseteq A_T(x)$ for $i\in\{1,\ldots,s\}$, it holds that $|\bigcup_{i=1}^sV(C_i)|\geq p+1$.
Let $S=\alpha(A_T(x))$. Note that $|S|\leq d+1\leq q$. Consider $A=S\cup\bigcup_{i=1}^s V(C_i)$ and $B=(V(G)\setminus A)\cup S$. We obtain that $(A,B)$ is a separation of $G$ of order at most $q$ and $|A\setminus B|\geq p+1$. Since  $G$ is $(p,q)$-unbreakable, we have that $|B|\leq p+q$. 
Notice that $X\subseteq B$. Thus, $|X|\leq p+q$. 

To show the second claim, note that $|L(T)|\geq p+q$.  In the same way as above,  there is $x\in L(T)$ such that for the components $C_1,\ldots,C_s$ of $G-X$ with $N_G(V(C_i))\subseteq A_T(x)$ for $i\in\{1,\ldots,s\}$, it holds that $|\bigcup_{i=1}^sV(C_i)|\geq p+1$. We show that there is $i\in\{1,\ldots,s\}$ such that $|V(C_i)|\geq p+1$. For the sake of contradiction, assume that $|V(C_i)|\leq p+1$ for all  $i\in\{1,\ldots,s\}$.  
Then there is a set of indices $I\subseteq\{1,\ldots,s\}$ such that $|\bigcup_{i\in I}V(C_i)|\geq p+1$ and for every proper $I'\subset I$,
$|\bigcup_{i\in I'}V(C_i)|\leq p$. Because each component has at most $p$ vertices, $|\bigcup_{i\in I}V(C_i)|\leq 2p$. Consider $A=S\cup\bigcup_{i=1}^sV(C_i)$, where $S=\alpha(A_T(x))$ and $B=(V(G)\setminus A)\cup S$. Note that $|B\setminus S|\geq |V(G)|-2p-q\geq p+1$. Then $(A,B)$ is a separation of $G$ of order at most $q$ with  $|A\setminus B|\geq p+1$ and $|B\setminus A|\geq p+1$; a contradiction with the condition that $G$ is $(p,q)$-unbreakable. This implies that there is a component $C$ of $G-X$ with $|V(C)|\geq p+1$. 

Because $G$ is a $(p,q)$-unbreakable graph and $|N_G(V(C))|\leq q$, we have that $|V(G)\setminus N_G[V(C)]|\leq p$. To see it, it is sufficient to consider the separation $(A,B)$ of $G$ with $A=N_G[V(C)]$ and $B=V(G)\setminus V(C)$. Clearly, $|B\setminus A|\leq p$ and, therefore, $|V(G)\setminus N_G[V(C)]|\leq p$.  This also implies the uniqueness of a component of $G-S$ with at least $p+1$, because for every other component $C'$, we have that $V(C')\subseteq B\setminus A$. This concludes the proof.
\end{proof}

Using the notation in Lemma~\ref{lem_es_bound}, we say that a component $C$ of $G-X$ with at least $p+1$ vertices is \emph{big} and the other components are \emph{small}.

We can use backtracking to verify, given a  $X$, whether $\dpt(X)\leq d$. For this we combine Lemmas~\ref{lem_nice} and~\ref{lem_depth-disconn} with backtracking and obtain the following straightforward lemma.
 
\begin{lemma}\label{lem_repr-find}
Given a graph $G$, a set of vertices $X\subseteq V(G)$, and an integer $d\geq -1$, it can be decided in 
$|X|^{\Oh(d)}\cdot n^{\Oh(1)}$ time whether $\dpt(X)\leq d$.
\end{lemma}

We have to solve \probED{\star} for $\star\in \{\edphione,\edphitwo\}$ on instances of bounded size. It is straightforward to see that this also can be done by backtracking following the definitions of $\edphi{\edphione}$ and $\edphi{\edphitwo}$. 

\begin{lemma}\label{lem_bounded-size}
Let $\varphi$ be a FOL formula. Then $\probED{\star}$ can be solved in $n^{\Oh{(k+|\varphi|)}}$ time   for $\star\in \{\edphione,\edphitwo\}$.
\end{lemma}

We also need the following technical lemma that will allow us to consider inclusion minimal elimination sets.

\begin{lemma}\label{lem_incl-min}
Let $G$ be a connected graph and let $X$ be a nonempty elimination set with a nice representation $(T,\alpha)$. Let also $C$ be a component of $G-X$ anchored in $x^*\in V(T)$. Suppose that  $S\subseteq N_G(V(C))$ and $C'$ is a component of $G-S$ with $V(C)\subseteq V(C')$.
Then  there is an elimination set $X'\subseteq X$ with a nice representation  $(T',\alpha')$ such that $V(T')\subseteq V(T)$ and the following is fulfilled:
\begin{itemize}
\item[(a)] $S\subseteq X'$ and $(N_G(V(C))\setminus S)\cap X'=\emptyset$, 
\item[(b)] for every component $H$ of $G-X'$, either $V(H)\subseteq V(C')$ or $H$ is a component of $G-X$, 
\item[(c)] for every node $y\in V(T')$, $\dpt_{T'}(y)\leq\dpt_T(y)$, 
\item[(d)] if a component $H$ of $G-X'$ distinct from $C'$ is anchored in a leaf $z$ of $T$, then $H$ in anchored in $z$ in $T'$ and $z$ is a leaf of $T'$, 
\item[(e)] if $x^*$ is  a leaf  of $T$ and $x^*\in S$, then $C'$ is anchored in $x^*$ in $T'$.
\end{itemize}  
\end{lemma}

\begin{proof}
Let $R=N_G(V(C))$.  The claim is  trivial if $S=\emptyset$, because $C'=G$ and we can take $X'=\emptyset$.  We assume that this is not the case. The proof is by induction on $|R\setminus S|$. The claim is straightforward  if $R=S$ as we can take $X'=X$ and consider the same representation $(T,\alpha)$.
The crucial case is the case $|R\setminus S|=1$. Let $u$ be the unique vertex of $R\setminus S$. We consider two possibilities for $u$. Let $v=\alpha(r)$, where $r$ is the root of $T$.

\medskip
\noindent{\bf Case~1.} $u=v$. Let $W=V(C')$. Notice that a vertex $w\in V(G)$ is in $W$ if and only if either $w\in V(C)$ or $w\notin S\cup V(C)$ and $G-S$ has a $(u,w)$-path. Define $X'=X\setminus W$. Clearly, (a) holds for this $X'$. Observe that for a component $H$ of $G-X$, we have that $V(H)\subseteq W$ if $N_G(V(H))$  contains a vertex of $W$ and $V(H)\cap W=\emptyset$ otherwise. In particular, this implies (b).

Next, we construct $T'$ and $\alpha'$. We set $V(T')=\alpha^{-1}(X')$ and define $\alpha'(x)=\alpha(x)$ for every $x\in V(T')$. Because $S\neq\emptyset$, there is a descendant  $r'$ of $r$ of minimum depth such that $\alpha(r')\in S$.
For every $w\in X'$ distinct from $\alpha(r')$, we consider $x=\alpha^{-1}(w)$ and find a proper ancestor $y$ of $x$ in $T$ of maximum depth such that $\alpha(y)\in X'$. Then we define $y$ be the parent of $x$. 

We argue that $T'$ is a tree rooted in $r'$. We have to show that for  every $w\in X'$ distinct from $\alpha(r')$, we have an ancestor $y$ of $x=\alpha^{-1}(w)$ in $T$ such that $\alpha(y)\in X'$. For the sake of contradiction, assume that there is $w\in X'$ such that for every proper ancestors $y$ of $x=\alpha^{-1}(w)$ in $T$, $\alpha(y)\notin X'$. Clearly, $x$ is not a descendant of $r'$ in $T$. In particular, $r'$ and $x$ are incomparable. 
Let $z$ be  the lowest proper ancestor of $r'$ and $y$ in $T$. We have that $\alpha(A_T(z))$ is an $(\alpha(r'),w)$-separator of $G$ and, moreover, $S$ has no vertices in the component $G-\alpha(A_T(z))$ containing $\alpha(x)$. Since $(T,\alpha)$ is nice, this component has an $(\alpha(z),w)$-path. Because $\alpha(z)\notin S$, $G-S$ has a $(u,\alpha(z))$-path. We conclude that $G-S$ has a $(u,w)$-path and $w\notin X'$; a contradiction.
This proves that $T'$ is a tree rooted in $r'$. 

We prove that $(T',\alpha')$ represents $X'$. Towards  a contradiction, assume that this is not the case, that is, there are distinct $x,y\in V(T')$ whose lowest common descendant $z\neq x,y$ and $\alpha'(x)$ and $\alpha'(y)$ are in the same component of $G-A_{T'}(z)$. By the definition of $T'$, $z$ has a descendant $z'$ such that $z'\neq x,y$ is the lowest common ancestor of $x$ and $y$ in $T$.  Clearly, either $x\notin N_G(V(C))$ or $y\notin N_G(V(C))$. By symmetry, assume that $y\notin N_G(V(C))$. Because $(T,\alpha)$ is a representation of $X$, $\alpha(A_T(z'))$ is an $(\alpha(x),\alpha(y))$-separator. This means that every $(\alpha'(x),\alpha'(y))$-path in $G$ contains a vertex of $(\alpha(x),\alpha(y))$. In particular, this implies that there is a vertex $z''\in A_T(z')$ such that $G$ has an $(\alpha(z''),\alpha'(y))$-path $P$ in $G$ such that the internal vertices of the path are in the component $G-\alpha(A_T(z'))$ containing $\alpha'(y)$. Because $y\notin N_G(V(C))$, we have that $P$ avoids the vertices of $S$. Since $z'\notin X'$, $G-S$ has a $(u,\alpha(z''))$-path $P'$. Concatenating $P'$ and $P$, we obtain that $G-S$ has a $(u,\alpha(y'))$-path. However this contradicts that $\alpha'(y)\in X'$. This proves that $(T',\alpha')$ represents $X'$.

By the construction of $T'$, it is easy to see that $T'$ is nice, because $T$ is nice. Also the construction of $T'$ immediately implies (c)--(e). This concludes the analysis of the first case.

\medskip
\noindent{\bf Case~2.} $u\neq v$. We show the claim by induction on $d=\dpt(T)$. Notice that $\dpt(T)\geq 1$, because $X\setminus\{v\}\neq\emptyset$. Let $C_1,\ldots,C_s$ be the components of $G-v$ such that $X_i=X\cap V(C_i)\neq\emptyset$ for $i\in\{1,\ldots,s\}$. Because $T$ is nice, $T$ has children $r_1,\ldots,r_s$ such that for every $i\in\{1,\ldots,s\}$, the subtree $T_i$ of $T$ rooted in $r_i$ together with $\alpha_i(x)=\alpha(x)$ for $x\in V(T_i)$ represent $X_i$ in $C_i$. By Observation~\ref{obs:anchor}, we can assume without loss of generality that $V(C)\subseteq V(C_1)$ and $N_G(V(C))\subseteq V(C_1)\cup\{v\}$.  
Because $\dpt(T_1)<\dpt(T)$, we can apply the inductive assumption and construct an elimination set $X_1'\subset X_1$ with a nice representation $(T_1',\alpha_1')$ satisfying (a)--(e). Then we construct $X'=X_1'\cup\bigcup_{i=1}^s X_i$. Then we construct $T'$ from $T_1'$ and $T_2,\ldots,T_s$ by making $r_1'$ and $r_2,\ldots,r_s$ the children of $r$, where $r_1'$ is the root of $T_1'$.  We set 
\begin{equation*}
\alpha'(x)=
\begin{cases}
\alpha_i(x),&\mbox{if }x\in X_i\text{ for some }i\in\{2,\ldots,s\},\\
\alpha_1'(x),&\mbox{if }x\in X_1'. 
\end{cases}
\end{equation*}
It is straightforward to verify that $X'$ and $(T',\alpha')$ satisfy (a)--(e).

\medskip
This concludes the poof for the base case $|R\setminus S|=1$. To show the claim for $|R\setminus S|>1$, consider a vertex $w\in R\setminus S$ and apply the claim for $S'=S\cup\{w\}$ using the inductive assumption. We have that there is 
an elimination set $X'\subseteq X$ with a nice representation  $(T',\alpha')$ such that $V(T')\subseteq V(T)$ and (a)--(e) are fulfilled with respect to $S'$. Then we apply the claim for $X'$ and $(T',\alpha')$ with respect to the component $C'$ and $S$. Clearly, we obtain an elimination set $X''\subseteq X'\subseteq X$ with a nice representation  $(T'',\alpha'')$ such that $V(T'')\subseteq V(T')\subseteq V(T)$ satisfying (a)--(e). This completes the proof.
\end{proof}

In our algorithms, we use the \emph{random separation} technique introduces by  Cai, Chan and Chan in~\cite{CaiCC06}. To avoid dealing with randomized algorithms, we use the following lemma stated by Chitnis et al. in~\cite{ChitnisCHPP16}.

\begin{lemma}[\cite{ChitnisCHPP16}]\label{lem_derand}
Given a set $U$ of size $n$ and integers $0\leq a,b\leq n$, one can construct in time $2^{\Oh(\min\{a,b\}\log (a+b))}\cdot n\log n$ a family $\mathcal{F}$ of at most   $2^{\Oh(\min\{a,b\}\log (a+b))}\cdot \log n$ subsets of $U$ such that the following holds: for any sets $A,B\subseteq U$, $A\cap B=\emptyset$, $|A|\leq a$, $|B|\leq b$, there exists a set $R\in \mathcal{F}$ with $A\subseteq R$ and $B\cap R=\emptyset$. 
\end{lemma}

\subsection{Algorithm for unbreakable graphs}\label{sec_unbreak}
In this subsection, we give \classFPT-algorithms for \probED{\star} for $\star\in\{\edphione,\edphitwo,\edphthree\}$ for FOL formulas $\varphi\in\Sigma_3$ on unbreakable graphs.   
Throughout the subsection, we assume without loss of generality that 
\begin{equation*}
\varphi=\exists x_1\cdots\exists x_r\forall y_1\cdots \forall y_s \exists z_1\cdots\exists z_t~\chi,
\end{equation*}
where $\chi$ is quantifier-free and $r,s,t$ are positive integers, because we always can write a FOL formula from $\Sigma_3$ in this form by adding dummy variables if necessary.  We also write $\mathbf{x}=\langle x_1,\ldots,x_r\rangle$, $\mathbf{y}=\langle y_1,\ldots,y_s\rangle$, and
$\mathbf{z}=\langle z_1,\ldots,z_t\rangle$.

Notice that $\edphi{\edphione}(G)=0$ if and only if for every component $C$ of $G$, $C\models\varphi$. Also $\edphi{\edphitwo}(G)=0$ ($\edphi{\edphthree}(G)=0$, respectively) if and only if $G\models\varphi$. This implies that \probED{\star} for $\star\in\{\edphione,\edphitwo,\edphthree\}$ can be solved in time $n^{\Oh(|\varphi|)}$ if $k=0$ by Observation~\ref{obs:mod-check}, that is,
Theorem~\ref{thm_algorithm_informal} trivially holds for $k=0$. Hence, throughout this subsection we assume that the parameter $k$ in the considered instances is positive. 

By Theorem~\ref{thm_meta} and Lemma~\ref{lem_expr}, to prove Theorem~\ref{thm_algorithm_informal}, it is sufficient to demonstrate FPT algorithm for the considered problems on $(p(k),k)$-unbreakable graphs for a computable function $p\colon \mathbb{N}\rightarrow\mathbb{N}$.  Slightly abusing notation, we write $p$ instead of $p(k)$.  

The algorithms for \probED{\star} for $\star\in\{\edphione,\edphitwo,\edphthree\}$ are similar. However, there are differences that make it inconvenient to describe them together. Hence, we first give the details of the algorithm for \probED{\edphione} and then more briefly explain our algorithms for \probED{\edphitwo} and \probED{\edphthree}. 
Then we derive Theorem~\ref{thm_algorithm_informal} from Lemmas~\ref{lem_alg-one}, \ref{lem_alg-two}, and \ref{lem_alg-three} in which we summarize the properties of the algorithms for the considered problems. 

\paragraph{Algorithm for \probED{\edphione}.} 
Let $(G,k)$ be an instance of \probED{\edphione}, where $G$ is a $(p,k)$-unbreakable graph.
We assume without loss of generality that $G$ is connected. Otherwise, because $\edphi{\edphione}(G)=\max\{\edphi{\edphione}(C)\mid C\text{ a component of }G\}$, we can solve the problem for each component separately. If $|V(G)|\leq (3p+2k)(p+1)$, we solve the problem in $(p+k)^{\Oh(k+|\varphi|)}$ time by 
Lemma~ \ref{lem_bounded-size}. From now we assume that  $|V(G)|> (3p+2k)(p+1)$.

By Lemma~\ref{lem_elim-dist-one}, 
$(G,k)$ is a yes-instance of \probED{\edphione} if and only if  $G$ contains an elimination set $X$ of depth at most $k-1$ such that $C\models \varphi$ for every component $C$ of $G-X$. Our algorithm finds such a set $X$, called a \emph{solution}, if it exists. We verify in $n^{\Oh(|\varphi|)}$ time whether $X=\emptyset$ has the required property and return {\sf yes} if this holds. Assume that this is not the case, that is, we have to find a nonempty solution. 

Suppose that $(G,k)$ is a yes-instance and let $X$ be a solution with a representation $(T,\alpha)$. 
 By Lemma~\ref{lem_es_bound}, $|X|\leq p+k$ and   there is a unique big component $C$ of $G-X$ with at least $p+1$ vertices, the other components are small,  and $|V(G)\setminus N_G[V(C)]|\leq p$.  By Observation~\ref{obs:anchor}, 
$N_G(V(C))\subseteq \alpha(A_T(x))$, where $x$ is an anchor of $C$. In particular, this means that $|N_G(V(C))|\leq k$. We use these properties to identify $C$. This is done by combining the random separation technique~\cite{CaiCC06} with a recursive branching algorithm.

We use random separation to highlight  the hypothetical sets $S=N_G(V(C))$ and $U=V(G)\setminus N_G[V(C)]$ (if they exist). To avoid randomized algorithms, we directly use the derandomization tool from Lemma~\ref{lem_derand}. By this lemma, we can construct in $2^{\Oh(\min\{p,k\}\log(p+k))}\cdot n\log n$ time a family $\mathcal{F}$ of at most $2^{\Oh(\min\{p,k\}\log(p+k))}\cdot \log n$ subsets of $V(G)$ such that there is $R\in\mathcal{F}$ such that $U\subseteq R$ and $S\cap R=\emptyset$. In our algorithm, we go over all sets $R\in  \mathcal{F}$ and for each set $R$, we check whether there is a solution $X$ such that  $U\subseteq R$ and $S\cap R=\emptyset$ for the sets $S$
  and $U$ corresponding to $X$ (recall that  $S=N_G(V(C))$ and $U=V(G)\setminus N_G[V(C)],$ where $C$ is the unique big component of $G-X$ with at least $p+1$ vertices). Clearly, $(G,k)$ is a yes-instance of \probED{\edphione} if and only if there is  $R\subseteq \mathcal{F}$ and a solution $X$ with the required property. 
  
  \begin{wrapfigure}{r}{55mm} 
 \vspace{-4mm}	\vspace{-4mm}
	\centering{\includegraphics[scale=0.5]{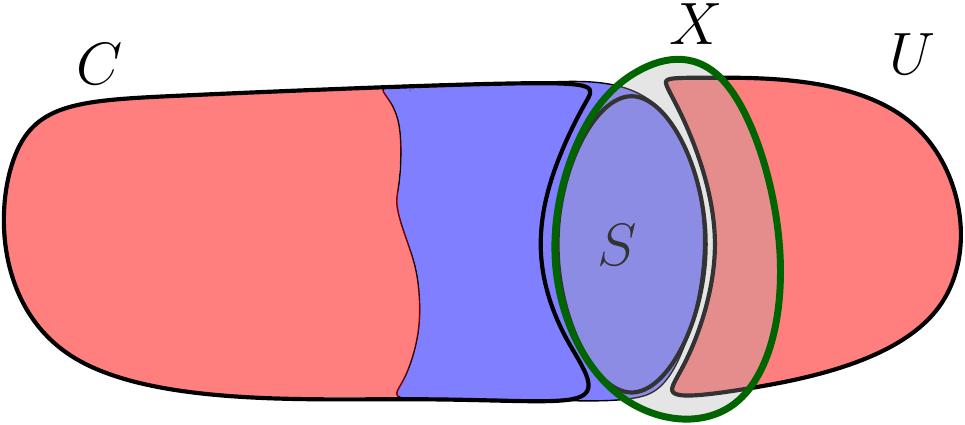}}
	\vspace{-4mm}
	\caption{A visualization of the set $X$, the component $C$, the sets $S$, and $U$, and and the way the red and blue colors are distributed among them.
	}
	\vspace{-4mm}
	\label{label_sysmbolization}
\end{wrapfigure}
  
From now on we assume that   $R\in \mathcal{F}$ is given. We set $B=V(G)\setminus R$. We say that the vertices of $R$ are \emph{red}  and the vertices of $B$ are \emph{blue}. We also call the components of $G[R]$  \emph{red components} of $G$ and we use the same convention for induced subgraphs of $G$. A solution $X$ is \emph{colorful} if the vertices of $U$ are red and the vertices of $S$ are blue (see \autoref{label_sysmbolization}). 
The crucial property of colorful solutions is that
\begin{quote}
 if a red vertex $v$ is in $U$, then the set of vertices of the red component $H$ containing $v$ is a subset of $U$. 
\end{quote}

If $G-X$ has a big component $C$ and $C\models \varphi$, then there is an $r$-tuple $\mathbf{v}=\langle v_1,\ldots,v_r\rangle$ of  vertices of $C$ such that $(C,\mathbf{v})\models \varphi[\mathbf{x}]$ (recall that $\varphi[\mathbf{x}]$ is the formula with the free variables $x_1,\ldots,x_r$ obtained from $\varphi$ by the removal of the quantification over $x_1,\ldots,x_r$, and  $(C,\mathbf{v})\models \varphi[\mathbf{x}]$ means that  $\varphi[\mathbf{x}]$ evaluates \emph{true} on $G$ when $x_i$ is assigned $v_i$ for all $i\in\{1,\ldots,r\}$). Using brute force, we consider all $r$-tuples $\mathbf{v}=\langle v_1,\ldots,v_r\rangle$ of  vertices of $G$, and for each $\mathbf{v}$, we explain how to check whether there is a colorful solution $X$ with the big component $C$ such that $v_i\in V(C)$ for all $i\in\{1,\ldots,r\}$.  Note that at most $n^r$ $r$-tuples $\mathbf{v}$ can be listed in $n^{\Oh(|\varphi|)}$ time. The algorithm returns {\sf yes} if we find a colorful solution for some choice of $\mathbf{v}$, and it concludes that there is no colorful solution for the considered selection of $R$ otherwise.  
  
From now we assume that $\mathbf{v}=\langle v_1,\ldots,v_r \rangle$ is fixed. Because these vertices should be in $C$, we temporarily (i.e., only for the current choice of $\mathbf{v}$) recolor them red to simplify further notation. 
We apply a recursive branching algorithm to find $C$ and $S$.  

By definition, we have that $(C,\mathbf{v})\models\varphi[\mathbf{x}]$ if and only if for every $s$-tuple $\mathbf{u}=\langle u_1,\ldots,u_s\rangle$ of  vertices of $C$,
$(C,\mathbf{vu})\models \varphi[\mathbf{xy}]$.  
Suppose that $(C,\mathbf{v})\not\models\varphi[\mathbf{x}]$. Then there is an $s$-tuple $\mathbf{u}=\langle u_1,\ldots,u_s\rangle$ of vertices such that 
$(C,\mathbf{vu})\not\models \varphi[\mathbf{xy}]$.  Notice now that, because $\varphi\in\Sigma_3$, we have that for any induced subgraph $C'$ of $C$ 
such that $v_i\in V(C')$ for every $i\in\{1,\ldots,r\}$,  if $(C',\mathbf{v})\models\varphi[\mathbf{x}]$, then there is $j\in\{1,\ldots,s\}$ such that $u_j\notin V(C')$. This implies that if  
 $(C,\mathbf{vu})\not\models \varphi[\mathbf{xy}]$, then there is $j\in\{1,\ldots,s\}$ such that either $u_j\in S$ and should be deleted or $u_j$ is in a component of $G-S$ distinct from $C$ and this component should be deleted together with its neighborhood. Note that $u_j$ is blue in the first case, and $u_j$ is red in the second. Moreover, in the second case, we should delete the red component containing $u_j$ together with its blue neighborhood. We branch on all possible deletions of $v_{i}$'s,  using the following subroutine $\textsc{FindC}(C,S,h)$, where  we initially set $C:=G$, $S:=\emptyset$, and $h:=k$.

\medskip
\noindent
{\bf Subroutine} $\textsc{FindC}(C,S,h)$. 
\begin{itemize}
\item If $(C,\mathbf{v})\models\varphi[\mathbf{x}]$ and $h\geq 0$, then return $C$, $S$, and stop.
\item If $(C,\mathbf{v})\not\models\varphi[\mathbf{x}]$ and $h\leq 0$, then stop.
\item If  $h\geq 1$ and there is an $s$-tuple $\mathbf{u}=\langle u_1,\ldots,u_s\rangle$ of vertices of $C$ such that $(C,\mathbf{vu})\not\models\varphi[\mathbf{xy}]$, then do the following for every $j\in\{1,\ldots,s\}$.
\begin{itemize} 
\item If $u_j\in B$ and there is a component $C'$ of $C-u_j$ such that $v_i\in V(C')$ for all $i\in\{1,\ldots,r\}$, then call $\textsc{FindC}(C',S\cup\{u_j\},h-1)$.
\item If $u_j\in R$ and there is a red component $H$ of $C$ with the set of vertices $W$ and $S'=N_C(W)$ such that (a) $u_j\in W$, (b) $|S'|\leq h$, and (c) there is a component $C'$ of $C-N_C[W]$ with  $v_i\in V(C')$ for all $i\in\{1,\ldots,r\}$, then  call $\textsc{FindC}(C',S\cup S', h-|S'|)$.
 \end{itemize}
\end{itemize}

We show the following lemma.

\begin{lemma}\label{lem_correct_one}
If $X$ is an inclusion minimal colorful solution to $(G,k)$ with the big component $C$ such that $v_i\in V(C)$ for all $i\in\{1,\ldots,s\}$ and  $(C,\mathbf{v})\models\varphi[\mathbf{x}]$, then there is a leaf of the search tree produced by  $\textsc{FindC}(G,\emptyset,k)$ for which the subroutine outputs $C$ and $S=N_G(V(C))$.
\end{lemma} 
 
\begin{proof} 
To prove the lemma, we show the following claim. If the subroutine  $\textsc{FindC}$ is called for $(\tilde{C},\tilde{S},\tilde{h})$ such that 
(a) $V(C)\subseteq V(\tilde{C})$, (b) $\tilde{S}=N_G(V(\tilde{C}))$, (c) $\tilde{S}\subseteq S$, and (d) $\tilde{h}=k-|\tilde{S}|$, then either the subroutine outputs $\tilde{C}$ and $\tilde{S}$ or it recursively calls  $\textsc{FindC}(\tilde{C}',\tilde{S}',\tilde{h}')$, where 
 (a$'$) $V(C)\subseteq V(\tilde{C}')$, (b$'$) $\tilde{S}=N_G(V(\tilde{C}'))$, (c$'$) $\tilde{S}'\subseteq S$, and (d$'$) $\tilde{h}'=k-|\tilde{S}'|$. 
 
 Notice that $\tilde{h}=k-|\tilde{S}|\geq 0$, because $|S|\leq k$. Hence, if  $(\tilde{C},\mathbf{v})\models\varphi[\mathbf{x}]$, then   $\textsc{FindC}(\tilde{C},\tilde{S},\tilde{h})$ outputs  $\tilde{C}$ and $\tilde{S}$ in the first step, and the claim holds. Assume that
 $(\tilde{C},\mathbf{v})\not\models\varphi[\mathbf{x}]$.   Because the subroutine is called  only for connected induced subgraphs of $G$, we have that 
 $\tilde{S}\subset S$ and, therefore, $\tilde{h}>0$. This implies that the subroutine does not stop in the second step. Then it proceeds to the third step and finds  an $s$-tuple $\mathbf{u}=\langle u_1,\ldots,u_s\rangle$ of vertices of $\tilde{C}$ such that $(\tilde{C},\mathbf{vu})\not\models\varphi[\mathbf{xy}]$. Because $(C,\mathbf{v})\models\varphi[\mathbf{x}]$, there is a $j\in \{1,\ldots,s\}$ such that $u_j\notin V(C)$. We consider the following two cases. 
 
\medskip
\noindent
{\bf Case~1.} $u_j\in S$. Notice that because $X$ is a colorful solution, $u_j$ is blue in this case. Observe also that $\tilde{C}-u_j$ has a component $\tilde{C}'$ such that $V(C)\subseteq V(\tilde{C}')$. 
Then the subroutine calls $\textsc{FindC}(\tilde{C}',\tilde{S}',\tilde{h}')$, where $\tilde{S}'=\tilde{S}\cup\{u_j\}$ and $\tilde{h}'=\tilde{h}-1$. It is easy to see that (a$'$)--(d$'$) are fulfilled for $\tilde{C}'$, $\tilde{S}'$, and $\tilde{h}'$. 

 \medskip
\noindent
{\bf Case~2.} $u_j\in U$. As $X$ is colorful, $u_j$ is red in this case. Let $H$ be the red component of $\tilde{C}$ containing $u_j$ and let $W=V(H)$. Because $X$ is a colorful solution, we have that $W\subseteq U$ and $N_{\tilde{C}}(W)\subseteq S$. Then $G-N_G[W]$ has a component $\tilde{C}'$ such that $V(C)\subseteq V(\tilde{C}')$. Then the subroutine calls $\textsc{FindC}(\tilde{C}',\tilde{S}',\tilde{h}')$, where $\tilde{S}'=\tilde{S}\cup N_{\tilde{C}}(W)$ and $\tilde{h}'=\tilde{h}-|N_{\tilde{C}}(W)|$. We obtain that (a$'$)--(d$'$) are fulfilled for $\tilde{C}'$, $\tilde{S}'$, and $\tilde{h}'$. This concludes the case analysis and the proof of the claim.

\medskip
Observe that conditions (a)--(d) of the claim are fulfilled if $C=G$, $S=\emptyset$, and $h=k$. Then the inductive application of the claim proves that there is a leaf of the search tree for which it outputs $\tilde{C}$ and  $\tilde{S}$ such that 
(a) $V(C)\subseteq V(\tilde{C})$, (b) $\tilde{S}=N_G(V(\tilde{C}))$, and (c) $\tilde{S}\subseteq S$. Recall that $X$ is an inclusion minimal colorful solution. Then Lemma~\ref{lem_incl-min} immediately implies that $C=\tilde{C}$ and $S=\tilde{S}$ and this concludes the proof.
 \end{proof} 
 
 Note that the number of branches of every node of the search tree produced by  $\textsc{FindC}(G,\emptyset,k)$ is at most $s$ and the depth of the search tree is at most $k$. This implies that the search tree has at most $s^k$ leaves.  By Lemma~\ref{lem_correct_one}, if $(G,k)$ has an inclusion minimal colorful solution $X$, then the subroutine outputs the corresponding big component $C$ containing $v_1,\ldots,v_r$ and $S$. We consider all pairs $(C,S)$ produced by  $\textsc{FindC}(G,\emptyset,k)$ and for each of these pairs, we verify whether there is a colorful solution corresponding to it. If we find such a solution we return {\sf yes} (or return the solution), and we return {\sf no} if we fail to find a clorful solution for each $C$ and $S$.  In the last case we conclude that we have no colorful solution and discard the current choice of $R\in \mathcal{F}$. 
  
Assume that $C$ and $S$ are given. Recall that $v_i\in V(C)$ for $i\in\{1,\ldots,r\}$, $S=N_G(V(C))$, and $(G,\mathbf{v})\models \varphi[\mathbf{x}]$.
First, we check whether $C$ is a big components of $G-S$ by verifying whether $|V(C)|\geq p+1$. Clearly, if $|V(C)|\leq p$, $C$ cannot be a big component of $G-X$ for a solution $X$ and we discard the considered choice of $C$ and $S$. Assume that this is not the case, that is, $|V(C)|\geq p+1$. Then because $G$ is a $(p,k)$-unbreakable graph, we have that $|V(G)\setminus N_{G}[V(C)]|\leq p$. We use brute force and consider every subset $Y\subseteq V(G)\setminus N_G[V(C)]$ and then verify whether (i) $X=S\cup Y$ is an elimination set of depth at most $k-1$ and (ii) for every component $C'\neq C$ of $G-X$, $C'\models \varphi$. Note that checking (i) can be done by Lemma~\ref{lem_repr-find}
in $(k+p)^{\Oh(k)}\cdot n^{\Oh(1)}$ time and (ii) can be verified in $n^{\Oh(|\varphi|)}$ time by Observation~\ref{obs:mod-check}. If we find $X=S\cup Y$ satisfying (i) and (ii), then we conclude that $X$ is a solution and return {\sf yes}. Otherwise, if we fail to find such a set, we return {\sf no}.

This concludes the description of the algorithm for \probED{\edphione} and its correctness proof.  We summarize in the following lemma.

\begin{lemma}\label{lem_alg-one}  
\probED{\edphione} on $(p,k)$-unbreakable graphs for $\varphi\in\Sigma_3$ can be solved in $2^{\Oh((p+k)\log (p+k) )}\cdot n^{\Oh(|\varphi|)}$ time.
\end{lemma} 
 
\begin{proof} 
Since the correctness of the algorithm was already established, it remains to evaluate the total running time. 
Recall that if  $|V(G)|\leq (3p+2k)(p+1)$, then the problem is solved in $(p+k)^{\Oh(k+|\varphi|)}$ time. Otherwise, we construct $\mathcal{F}$ of size at most
$2^{\Oh(\min\{p,k\}\log(p+k))}\cdot \log n$ in $2^{\Oh(\min\{p,k\}\log(p+k))}\cdot n\log n$ time. Then for every $R\in \mathcal{F}$, we try to find a colorful solution. For this, we first guess $\mathbf{v}$. Clearly, we have $n^{\Oh(|\varphi|)}$ possibilities for the choice of $\mathbf{v}$. Then we run 
the subroutine $\textsc{FindC}(C,S,h)$. Note that the search tree produced by the subroutine has at most $|\varphi|^k$ leaves and each call (without recursive calls) requires $n^{\Oh(|\varphi|)}$ time. Then the running time of the subroutine is $|\varphi|^k\cdot n^{\Oh(|\varphi|)}$. We consider the pairs $(C,S)$ produced by the subroutine, and for each $C$ and $S$, we verify whether we have a corresponding colorful solution $X$. The brute force selection of $X$ can be done in $2^{\Oh(p)}$ time. Then checking whether $X$ is a solution requires  $(k+p)^{\Oh(k)}\cdot n^{\Oh(1)}$. Then we conclude that the total running time is $2^{\Oh((p+k)\log (p+k) )}\cdot n^{\Oh(|\varphi|)}$.
\end{proof}

\paragraph{Algorithm for \probED{\edphitwo}.}  Let $(G,k)$ be an instance of \probED{\edphitwo}, where $G$ is a $(p,k)$-unbreakable graph. We check whether $G\models \varphi$ and immediately return {\sf yes} if this is fulfilled. Assume that this is not the case and that $\edphi{\edphitwo}(G)\geq 1$. Then we can assume without loss of generality that $G$ is connected.
Otherwise, because $\edphi{\edphitwo}(G)=\max\{1,\max\{\edphi{\edphitwo}(C)\mid C\text{ a component of }G\}\}$, we can solve the problem for each component separately. In the same way as with \probED{\edphione},  we solve the problem in $(p+k)^{\Oh(k+|\varphi|)}$ time by 
Lemma~ \ref{lem_bounded-size} if $|V(G)|\leq (3p+2k)(p+1)$. Therefore, from now on, we may assume that  $|V(G)|> (3p+2k)(p+1)$.

Let $(T,\alpha)$ be a representation of an elimination set $X$. Recall that 
 $\property_x$ denotes the set of components of $G-X$ anchored in $x$, where $x$ is a node of $T$. Also $G_x$ denotes the subgraph of $G$ induced by the vertices of the graphs of $\property_x$, that is, $G_x$ is the union of the components of $G-X$ anchored in $x$. 
By Lemma~\ref{lem_elim-dist-two}, 
$(G,k)$ is a yes-instance of \probED{\edphitwo} if and only if  $G$ contains an elimination set $X$ of depth at most $k-1$ with 
a representation $(T,\alpha)$ such that 
\begin{itemize}
\item[(i)] for every nonleaf node $x\in V(T)$, $C\models \varphi$ for every  $C\in\property_x$, 
\item[(ii)] for every leaf $x$ of $T$ with $\dpt_T(x)\leq k-2$, either $G_x\models \varphi$ or $C\models \varphi$ for every  $C\in\property_x$, and \item[(iii)] for every leaf $x$ of $T$ with $\dpt_T(x)=k-1$,  $G_x\models \varphi$.
\end{itemize}
We call such a set $X$ a \emph{solution}. 
 We observe that, given a set $X$, we can decide whether $X$ is a solution.

\begin{lemma}\label{lem_compute-repr}
Let $X\subseteq V(G)$ be nonempty. It can be decided in $|X|^k\cdot n^{\Oh(|\varphi|)}$ time whether $X$ has a representation $(T,\alpha)$ satisfying (i)--(iii).
\end{lemma}

\begin{proof}
Because $G$ is connected and $k\geq 1$, it is sufficient to verify the existence of a nice representation. We do it by a recursive algorithm that for a given $x\in X$ finds a nice representation $(T,\alpha)$ such that $\alpha(r)=x$, where $r$ is the root of $T$. More precisely, given a graph $G$, a nonempty $X\subseteq V(H)$, a vertex $x\in X$, and a positive integer $k$, the algorithm find a nice representation $(T,\alpha)$ of $X$ satisfying (i)--(iii) such that $\alpha(r)=x$ if such a representation exists. 

Suppose that $|X|=1$, that is, $X=\{x\}$.
If  $G-x\models\varphi$, the algorithm returns a single-vertex tree rooted in $r$ with $\alpha(r)=x$. 
If $G-x\not\models\varphi$ and $k\geq 2$, we check whether $C\models \varphi$ for every component $C-x$. If this holds, then again, the algorithm returns a single-vertex tree rooted in $r$ with $\alpha(r)=x$. In all other cases, the algorithm returns {\sf no}.

Suppose from now that $|X|\geq 2$. If $k=1$, then we immediately return {\sf no} and stop. Also if there is a component $C$ of $G-x$ such that  
$V(C)\cap X=\emptyset$ and $C\not\models \varphi$, the the algorithm returns {\sf no} and stops. Assume that these are not cases. Let $C_1,\ldots,C_s$ be the components of $G-x$ such that $X_i=X\cap V(C_i)\neq \emptyset$. 

For every $i\in\{1,\ldots,s\}$, we call the algorithm recursively for $C_i$, $X_i$, every $y\in X_i$, and $k-1$. If there is $i\in\{1,\ldots,s\}$ such that the algorithm failed to produce a representation for every choice of $y\in X_i$, the algorithm returns {\sf no} and stops. Otherwise, the algorithm 
finds for every $i\in\{1,\ldots,s\}$ a vertex $x_i\in X_i$ and a nice representation $(T_i,\alpha_i)$ of $X_i$ in $C_i$ satisfying (i)--(iii) (with respect to the new parameters) such that the root $r_i$ is mapped to $x_i$ by $\alpha_i$. We construct $T$ from $T_1,\ldots,T_s$ by creating a root $r$ and making it the parent of $r_1,\ldots,r_s$. Then
\begin{equation*}
\alpha(z)=
\begin{cases}
x,&\mbox{if }z=r,\\
\alpha_i(z),&\mbox{if }z\in V(T_i)\text{ for some }i\in\{1,\ldots,s\}.
\end{cases}
\end{equation*}
 
This completes the description of the algorithm. It is straightforward to verify its  correctness using the definition of a nice representation  of an elimination set. To decide whether $X$ has a representation $(T,\alpha)$ satisfying (i)--(iii), we run the algorithm for all $x\in X$. Clearly, a representation exists if and only if the algorithm produces a representation for some choice of $x$. Since in each call of the algorithm, we make at most $|X|$ recursive calls and the depth of the recursion is at most $k$, the total running time is $|X|^k\cdot n^{\Oh(|\varphi|)}$. 
\end{proof}

Suppose that $(G,k)$ is a yes-instance and let $X$ be a solution with a nice representation $(T,\alpha)$. 
 By Lemma~\ref{lem_es_bound}, $|X|\leq p+k$ and   there is a unique big component $C$ of $G-X$ with at least $p+1$ vertices, the other components are small,  and $|V(G)\setminus N_G[V(C)]|\leq p$.  By Observation~\ref{obs:anchor}, 
$N_G(V(C))\subseteq \alpha(A_T(x))$, where $x$ is an anchor of $C$. In particular, this means that $|N_G(V(C))|\leq k$. 
As with the algorithm for \probED{\edphione}, our aim is to identify $C$. We consider two possibilities for $C$.

First, we try to find $C$ assuming that one of the following holds: either (a) the anchor of $C$ is not a leaf of $T$ or (b) the anchor $x$ is leaf but $\dpt_T(x)<k-1$ and $C'\models \varphi$ for every  $C'\in\property_x$, or (c) $G_x=C$. In this case, the algorithm is essentially identical to the algorithm for \probED{\edphione}. We use Lemma~\ref{lem_derand} to highlight $S$ and $U=V(G)\setminus N_G[V(C)]$. Then we guess $\mathbf{v}$ in $C$
and call the subroutine $\textsc{FindC}(G,\emptyset,k)$ to enumerate all candidate big components $C$ and $S=N_G(V(S))$. The difference occurs only in the last step of the algorithm, where we find a solution $X$.  
We use brute force and consider every subset $Y\subseteq V(G)\setminus N_G[V(C)]$ and then verify whether $X=S\cup Y$ is an elimination set of depth at most $k-1$ satisfying (i)--(iii) using Lemma~\ref{lem_compute-repr}.
 If we find a required $X$, then we conclude that $X$ is a solution and return {\sf yes}. Otherwise, if we fail to find such a set for every candidate $C$, we return {\sf no} for the considered set $R$ and discard it. The correctness is proved and the running time is analysed in exactly the same way as for \probED{\edphione}.

Next, if we failed to find a solution so far, we consider the remaining possibility that the anchor $x$ of $C$ is a leaf of $T$ and $G_x\models\varphi$, where $G_x$ is a disconnected graph. Our algorithm for this case uses the same approach as the algorithm for \probED{\edphione} but the arguments are more involved, as we aim to identify  $C$ together with the other components of $G_x$. In other words,  we find $G_x$.

Let $S=N_G(V(G_x))$. Note that $S\subseteq \alpha(A_T(x))$ and, therefore, $|S|\leq k$. Observe that $N_G(V(C))\subseteq S$. 
 Let also $U=V(G)\setminus (V(C)\cup S)$. Because $C$ is a big component and $G$ is $(p,k)$-unbreakable, $|U|\leq p$. 

Similarly to the algorithm for \probED{\edphione}, we use  Lemma~\ref{lem_derand} to highlight hypothetical $S$ and $U$. By this lemma, we can construct in $2^{\Oh(\min\{p,k\}\log(p+k))}\cdot n\log n$ time a family $\mathcal{F}$ of at most $2^{\Oh(\min\{p,k\}\log(p+k))}\cdot \log n$ subsets of $V(G)$ such that there is $R\in\mathcal{F}$ such that $U\subseteq R$ and $S\cap R=\emptyset$. In our algorithm, we go over all sets $R\subseteq \mathcal{F}$ and for each set $R$, we check whether there is a solution $X$ such that  $U\subseteq R$ and $S\cap R=\emptyset$ for the sets $S$
  and $U$ corresponding to $X$. Clearly, $(G,k)$ is a yes-instance of \probED{\edphitwo} if and only if there is  $R\subseteq \mathcal{F}$ and a solution $X$ with the required property. 
  
From now on we assume that   $R\subseteq \mathcal{F}$ is given. We set $B=V(G)\setminus R$. In the same way as before, we say that the vertices of $R$ are \emph{red}  and the vertices of $B$ are \emph{blue}. The components of $G[R]$ are called \emph{red components} of $G$ and the same convention is used for induced subgraphs of $G$. A solution $X$ is called \emph{colorful} if the vertivces of $U$ are red and the vertices of $S$ are blue. We aim to find a colorful solution. 

Assume that a colorful solution $X$ exists.  
Suppose that $w=\alpha(x)$ for the leaf $x$ of $T$ that is the anchor of $G_x$. Notice that $w\in B$. Then for every component $C'$ of $G_x$ distinct from $C$, we have that $C'$ is a red component and $z\in N_G(V(C'))$. We also observe that by the assumption for $R$, if $C'$ is a red component of $G$ such that $w\in N_G(V(H))$, then either $V(C')\subseteq V(C)$ or $C'$ is a component of $G_x$ distinct from $C$.  
Using these observations, we consider all possible choices of $w$ in $B$, and decide whether there is a colorful solution $X$ such that for the required $G_x$, the leaf $x$ of $T$ is mapped to $w$. We say that $X$ is a colorful solution \emph{attached} to $w$. 

From now we assume that $w$ is given. Let 
$W=\bigcup V(H)$, where the union is taken over all red components $H$ of $G$ such that $w\in N_G(V(H))$. Notice that if there is a colorful solution $X$ attached to $w$ for the considered choice of $w$, then $W\subseteq V(G_x)$ for the corresponding graph $G_x$. 

 Since we require that $G_x\models \varphi$, then there is an $r$-tuple $\mathbf{v}=\langle v_1,\ldots,v_r\rangle$ of  vertices of $G_x$ such that $(G_x,\mathbf{v})\models \varphi[\mathbf{x}]$. Using brute force, we consider all $r$-tuples $\mathbf{v}=\langle v_1,\ldots,v_r\rangle$ of  vertices of $G$ distinct from $w$, and for each $\mathbf{v}$, we check whether there is a colorful solution $X$ with $G_x$ such that $v_i\in V(G_x)$ for all $i\in\{1,\ldots,r\}$.  Note that at most $n^r$ $r$-tuples $\mathbf{v}$ can be listed in $n^{\Oh(|\varphi|)}$ time. The algorithm returns {\sf yes} if we find a colorful solution attached to $w$ for some choice of $\mathbf{v}$, and it concludes that there is no colorful solution for the considered selection of $R$ otherwise.  
  
From this point we assume that $\mathbf{v}=\langle v_1,\ldots,v_r \rangle$ is fixed. Because these vertices should be in $G_x$, we temporarily (i.e., only for the current choice of $\mathbf{v}$) recolor them red to simplify further notation and recompute $W$ if necessary. 
We apply a recursive branching algorithm to find $F=G_x$ and $S$.  Similarly to the subroutine $\textsc{FindC}$, we construct the subroutine $\textsc{FindF}(F,S,h)$, where initially $F=G-w$, $S=\{w\}$, and $h=k-1$.

\medskip
\noindent
{\bf Subroutine} $\textsc{FindF}(F,S,h)$.
\begin{itemize}
\item If $(F,\mathbf{v})\models\varphi[\mathbf{x}]$ and $h\geq 0$, then return $F$, $S$, and stop.
\item If $(F,\mathbf{v})\not\models\varphi[\mathbf{x}]$ and $h\leq 0$, then stop.
\item If  $h\geq 1$ and there is an $s$-tuple $\mathbf{u}=\langle u_1,\ldots,u_s\rangle$ of vertices of $F$ such that $(F,\mathbf{vu})\not\models\varphi[\mathbf{xy}]$, then do the following for every $j\in\{1,\ldots,s\}$.
\begin{itemize} 
\item If $u_j\in B$ and there is an induced subgraph $F'$ of $F$ that is the disjoint union of the components of $F-u_j$ containing vertices of $W$ and vertices $v_i$ for $i\in\{1,\ldots,r\}$, then call $\textsc{FindF}(F',S\cup\{u_j\},h-1)$.
\item If $u_j\in R$ and there is a red component $H$ of $C$ with  
the set of vertices $Z$ and $S'=N_F[Z]$ such that (a) $u_j\in Z$, (b) $Z\cap W=\emptyset$ and $v_i\notin Z$ for all $i\in\{1,\ldots,r\}$, 
 (c) $|S'|\leq h$, and (d) there is an induced subgraph $F'$ of $F$ that is a disjoint union of the components of $F-N_F[W]$ containing  vertices of $W$ and  vertices $v_i$ for some $i\in\{1,\ldots,r\}$, then  call $\textsc{FindF}(F',S\cup S', h-|S'|)$.
 \end{itemize}
\end{itemize}

In the same way as with Lemma~\ref{lem_correct_one}, we show the following.

\begin{lemma}\label{lem_correct-two}
If $X$ is an inclusion minimal colorful solution attached to $w$ for $(G,k)$ with such that (a) $X$ has a representation $(T,\alpha)$ with $\alpha(x)=w$, (b) $W\subseteq V(G_x)$, (c) $v_i\in V(G_x)$ for all $i\in\{1,\ldots,s\}$ and  $(G_x,\mathbf{v})\models\varphi[\mathbf{x}]$, then there is a leaf of the search tree produced by  $\textsc{FindC}(G,\{w\},k-1)$ for which the subroutine outputs $F$ and $S=N_G(V(F))$.
\end{lemma} 
 
 Since the number of branches of every node of the search tree produced by  $\textsc{FindF}(G,\{w\},k-1)$ is at most $s$ and the depth of the search tree is at most $k$,  the search tree has at most $r^k$ leaves.  By Lemma~\ref{lem_correct-two}, if $(G,k)$ has an inclusion minimal colorful solution $X$ attached to $w$ with respect to some representation $(T,\alpha)$ of $X$, then the subroutine outputs the corresponding graph  $F=G_x$ containing $v_1,\ldots,v_r$ and $S$. We consider all pairs $(F,S)$ produced by  $\textsc{FindF}(G,\{w\},k-1)$ and for each of these pairs, we verify whether there is a colorful solution corresponding to it. If we find such a solution we return {\sf yes} (or return the solution), and we return {\sf no} if we fail to find a colorful solution for each $F$ and $S$.  In the last case we conclude that we have no colorful solution and discard the current choice of $R\in \mathcal{F}$. 
  
Assume that $F$ and $S$ are given. Recall that $v_i\in V(C)$ for $i\in\{1,\ldots,r\}$, $S=N_G(V(C))$, and $(G,\mathbf{v})\models \varphi[\mathbf{x}]$.
First, we check whether $F$ has a big components of $G-S$ by verifying whether $F$ has a component with at least $p+1$
 vertices. If we have no such a component,  we discard the considered choice of $F$ and $S$. Assume that this is not the case. Then because $G$ is a $(p,k)$-unbreakable graph, we have that $|V(F)\setminus N_{G}[V(F)]|\leq p$. We use brute force and consider every subset $Y\subseteq V(G)\setminus N_G[V(F)]$ and then verify whether $X=S\cup Y$ is a solution using Lemma~\ref{lem_compute-repr}.
 If we find a solution, we return {\sf yes}. Otherwise, if we fail to find $Y$ with the required properties, we return {\sf no}.
 
This concludes the description of the algorithm for \probED{\edphione} and its correctness proof.  We summarize in the following lemma that is proved in the same way as Lemma~\ref{lem_alg-one}.

\begin{lemma}\label{lem_alg-two}  
\probED{\edphitwo} on $(p,k)$-unbreakable graphs for $\varphi\in\Sigma_3$ can be solved in $2^{\Oh((p+k)\log (p+k) )}\cdot n^{\Oh(|\varphi|)}$ time.
\end{lemma} 

\paragraph{Algorithm for \probED{\edphthree}.} Our final task is to construct an algorithm for \probED{\edphthree}. 
Let $(G,k)$ be an instance of \probED{\edphthree}, where $G$ is a $(p,k)$-unbreakable graph. If $G\models \varphi$, then we return {\sf yes}. Assume that this is not the case and $\edphi{\edphthree}(G)\geq 1$. 

Suppose that $G$ is disconnected. Denote by $C_1,\ldots,C_s$ the components of $G$. Because $G$ is a $(p,k)$-unbreakable graph, at most one component can have more $p$ vertices. Then we can assume that $|V(C_i)|\leq p$ for every $i\in\{2,\ldots,s\}$. 
For each $i\in\{2,\ldots,s\}$, we solve \probED{\edphthree} for $(C_i,k-1)$ and $(C_i,k)$ in $2^p\cdot p^{\Oh(k+|\varphi|)}$ time using brute force.  
Let $i\in\{2,\ldots,s\}$. For each set $X\subseteq V(C_i)$, we check whether $\dpt(X)\leq k-2$ ($\dpt(X)\leq k-1$, respectively) applying Lemma~\ref{lem_repr-find},
and if this holds we verify whether  $G-X\models\varphi$. This can be done in $2^p\cdot p^{\Oh(k)}\cdot p^{\Oh(|\varphi|)}$ time. 
If we find that either there is $i\in\{2,\ldots,s\}$ such that $\edphi{\edphthree}(C_i)\geq k+1$ or 
there
are two distinct $i,j\in\{2,\ldots,s\}$ such that $\edphi{\edphthree}(C_i)=\edphi{\edphthree}(C_j)=k$, we return {\sf no} by Lemma~\ref{lem_depth-disconn}. If there is a unique $i\in\{2,\ldots,s\}$ with $\edphi{\edphthree}(C_i)=k$ and $\edphi{\edphthree}(C_j)\leq k-1$ for $j\in\{2,\ldots,s\}$ distinct from $i$, $(G,k)$ is a yes-instance if and only if $(C_1,k-1)$ is a yes-instance by Lemma~\ref{lem_depth-disconn}. If $\edphi{\edphthree}(C_i)\leq k-1$ for every $i\in\{2,\ldots,s\}$, then by the same lemma,  $(G,k)$ is a yes-instance if and only if $(C_1,k)$ is a yes-instance. Thus, we are able to reduce solving the problem on $G$ to solving it on $C_1$. This implies that we can assume without loss of generality that $G$ is connected.

If $|V(G)|\leq (3p+2k)(p+1)$, we again can solve the problem using brute force in $2^{(3p+2k)(p+1)}\cdot ((3p+2k)(p+1))^{\Oh(k+|\varphi|)}$  time
in the same way as above.
Then we assume that  $|V(G)|> (3p+2k)(p+1)$.

Given a subset $X\subseteq V(G)$, we can verify in $|X|^{\Oh(k)}\cdot n^{\Oh(1)}$ whether $\dpt(X)\leq k-1$ by Lemma~\ref{lem_repr-find} and then can check whether $G-X\models\varphi$ using Observation~\ref{obs:mod-check}. Based on this, we aim to find $X$ that we call a \emph{solution} in the same way as for the previously considered problems.

For \probED{\edphione}, we used the random separation technique to highlight a big component of $G-X$ (or rather its complement), and for \probED{\edphitwo}, besides a big component we had to highlight some specific small components composing $G_x$ together with the big component. Now we are highlighting the small components, $X$ and the neighborhood  $N_G(V(C))\subseteq X$ of the big component. 

Suppose that $(G,k)$ is a yes-instance with a solution $X$.
By Lemma~\ref{lem_es_bound}, $|X|\leq p+k$ and $|V(G)\setminus N_G[V(C)]|\leq p$, where $C$ is a big component of $G-X$. 
By Lemma~\ref{lem_derand}, we can construct the family $\mathcal{F}$ of subsets of $V(G)$ of size at most $2^{\Oh((p+k)\log(p+k))}\cdot \log n$ in 
$2^{\Oh((p+k)\log(p+k))}\cdot n\log n$ time such that if $(G,k)$ has a solution $X$, then $\mathcal{F}$ has a set $R$ such that $V(H)\subseteq R$ for every small component and $R\cap X=\emptyset$.  
Then for every $R\in \mathcal{F}$, we aim to find a solution $X$ such that the vertices of the small components of $G-X$ are in $R$ and $X\cap R=\emptyset$.

From this point  we assume that $R$ is given. Consider $U=V(G)\setminus R$.  If $C$ is a big component of a (hypothetical) solution $X$ satisfying the above conditions, then $N_G(V(C))\subseteq U$ and $|N_G(V(C))|\leq k$. Recall that $|X\setminus N_G(V(C))|\leq |V(G)\setminus N_G[V(C)]|\leq p$. Since $|U|\leq n$, applying Lemma~\ref{lem_derand} for $U$, we  construct the family $\mathcal{F}'$ of subsets of $U$ of size at most $2^{\Oh(\min\{p,k\}\log(p+k))}\cdot \log n$ in 
$2^{\Oh(\min\{p,k\}\log(p+k))}\cdot n\log n$ time such that   $\mathcal{F}'$ has a set $Y$ with the property that $X\setminus N_G(V(C))\subseteq Y$ and $N_G(V(C))\cap Y=\emptyset$.
We consider all $Y\in\mathcal{F}'$ and try find a solution $X$ such that 
\begin{itemize}
\item[(i)] the vertices of the small components of $G-X$ are in $R$,
\item[(ii)] for the big component $C$ of $G-X$, $X\setminus N_G(V(C))\subseteq Y$,
\item[(iii)] for the big component $C$ of $G-X$, $N_G(V(C))\subseteq B$, where $B=V(G)\setminus (R\cup Y)$.
\end{itemize}
If a solution $X$ satisfies (i)--(iii), then we say that $X$ is \emph{colorful}.

We say that  the vertices of $R$ are \emph{red}, the vertices of $Y$ are yellow, and the vertices of $B$ are \emph{blue}. The components of $G[R]$ are called \emph{red} and
the components of $G[R\cup Y]$ are called \emph{red-yellow} components of $G$, and we use the same term for the induced subgraphs of $G$.

Assume that $X$ is a colorful solution. Notice that if $H$ is a red component of $G$, then either $H$ is a small component of $G-X$ with $N_G(V(H))\subseteq X$ or $V(H)\subseteq V(C)$, where $C$ is the big component. Also we have that if $H$ is a red-yellow component of $G$, then either $V(H)\subseteq V(C)$ or every red component of $V(H)$ is a small component of $G-X$.
These are the crucial properties of colorful solutions exploited by our algorithm.

Because $G-X\models\varphi$ for a solution $X$, it should exist an $r$-tuple $\mathbf{v}=\langle v_1,\ldots,v_r\rangle$ of  vertices of $G-X$ such that $(G-X,\mathbf{v})\models \varphi[\mathbf{x}]$. In the same way as for the previous problem, we use brute force to list  all $r$-tuples $\mathbf{v}=\langle v_1,\ldots,v_r\rangle$ of  vertices of $G$. Then for each $\mathbf{v}$, we check whether there is a colorful solution $X$ such 
that $v_i\notin X$ for all  $i\in\{1,\ldots,r\}$ and $(G-X,\mathbf{v})\models \varphi[\mathbf{x}]$.
There are at most $n^r$ $r$-tuples $\mathbf{v}$ can be listed  $n^{\Oh(|\varphi|)}$ time. Our algorithm returns 
{\sf yes} if we find a colorful solution for some choice of $\mathbf{v}$, and it concludes that there is no colorful solution for the considered selection of $R$ otherwise.  

From now we assume that $\mathbf{v}=\langle v_1,\ldots,v_r \rangle$ is fixed. Again we observe that these vertices should not belong to $X$ and we recolor them red for the considered choice of $R$.
We use a recursive branching algorithm to find $X$.
The algorithm exploits  the subroutine $\textsc{FindX}(Z,h)$, where initially $Z=\emptyset$ and $h=p+k$.

\medskip
\noindent
{\bf Subroutine} $\textsc{FindX}(Z,h)$. 
\begin{enumerate}
\item Set $F:=G-Z$.
\item If $(F,\mathbf{v})\models\varphi[\mathbf{x}]$, $\dpt(Z)\leq k-1$, and $h\geq 0$, then return $Z$, and stop executing the algorithm.
\item If $(F,\mathbf{v})\not\models\varphi[\mathbf{x}]$ and $h\leq 0$, then stop executing the subroutine.
\item If  $h\geq 1$ and there is an $s$-tuple $\mathbf{u}=\langle u_1,\ldots,u_s\rangle$ of vertices of $F$ such that $(F,\mathbf{vu})\not\models\varphi[\mathbf{xy}]$, then do the following: 
\begin{itemize}
\item If $u_j\in R$ for every $j\in\{1,\ldots,s\}$, then stop executing the subroutine.
\item Otherwise, for every $j\in\{1,\ldots,s\}$ such that $u_j\in Y\cup B$, call $\textsc{FindX}(Z\cup\{u_j\},h-1)$. 
\end{itemize}
\item If $\dpt(Z)\geq k$, then for every $x\in Z$ such that there is a red-yellow component $H$ of $F$ with the properties (i) $|N_F(V(H))|\leq k$, (ii) $|V(H)|\leq p$,  
(iii) $x\in V(H)$, and (iv) $N_F[V(H)]\cap (B\cup Y)\neq \emptyset$,  set $S:=N_F[V(H)]\cap (B\cup Y)$ and  call $\textsc{FindX}(Z\cup S,h-|S|)$. 
\end{enumerate}

Notice that if the subroutine outputs $Z$, then we stop the algorithm and report that we found a solution. 
If we stop in other steps, then we only stop the execution of the subroutine for the current call.  
The crucial property of the subroutine are proved the following lemma. Since  $\textsc{FindX}(Z,h)$ substantially differs from $\textsc{FindC}(C,S,h)$ and $\textsc{FindF}(F,S,h)$, we provide the proof.

\begin{lemma}\label{lem_correct-three}
If $(G,k)$ has a colorful inclusion minimal solution $X$ with $v_i\in V(G)\setminus X$ for all $i\in\{1,\ldots,s\}$ such that $(G-X,\mathbf{v})\models\varphi[\mathbf{x}]$, then $\textsc{FindX}(\emptyset,p+k)$ returns $X$.
\end{lemma} 
 
\begin{proof} 
The lemma is proved similarly to Lemma~\ref{lem_correct_one}. Let $t=p+k$. We show that the algorithm maintains the following property: if the subroutine  $\textsc{FindX}$ is called for $(Z,h)$ such that 
(a) $Z\subseteq X$ and (b) $h=t-|Z|$, then either the subroutine outputs $Z=X$  or it recursively calls  $\textsc{FindX}(Z',h')$, where 
(a$'$) $Z'\subseteq X$ and (b$'$) $h'=t-|Z'|$

In the first step, the algorithms sets $F=G-Z$. 
If $(F,\mathbf{v})\models\varphi[\mathbf{x}]$, $\dpt(Z)\leq k-1$, $h\geq 0$, and $(F,\mathbf{v})\models\varphi[\mathbf{x}]$, then $Z$ is a colorful solution and the algorithm returns  
return $Z$. Since $X$ is inclusion minimal, we have that $X=Z$. Thus, the claim holds.
Assume that this is not the case. Since $Z\subseteq X$, we have that $h\geq 1$, that is, the subroutine does not stop in step~3. 
Clearly, we have that $(F,\mathbf{v})\not\models\varphi[\mathbf{x}]$ and/or $\dpt(Z)\geq k$.

Suppose that  $(F,\mathbf{v})\not\models\varphi[\mathbf{x}]$.  Then there is an $s$-tuple $\mathbf{u}=\langle u_1,\ldots,u_s\rangle$ of vertices of $F$ such that $(F,\mathbf{vu})\not\models\varphi[\mathbf{xy}]$. This means that the subroutine executes step~4. As 
$(G-X,\mathbf{v})\models\varphi[\mathbf{x}]$ and $Z\subseteq X$, there is $j\in\{1,\ldots,s\}$ such  that $u_j\in X\setminus X$. Because $X$ is a colorful solution, $u_j\in B$ . Therefore, the subroutine calls 
$\textsc{FindX}(Z',h')$ for $Z'=Y\setminus\{u_j\}$ and $h'=h-1$. It is easy to see that (a$'$) and (b$'$) are fulfilled.

Assume that $(F,\mathbf{v})\models\varphi[\mathbf{x}]$. Then $\dpt(Z)\geq k$. 
Because $\dpt(X)\leq k-1$, $X$ has a representation $(T,\alpha)$ with $\dpt(T)\leq k-1$. 
Because $\dpt(T)\leq k-1$ and $\dpt(Z)\geq k$, there are vertices $x,y\in Z$ such that the nodes $x'=\alpha^{-1}(x)$ and $y'=\alpha^{-1}(y)$ have the lowest common ancestor $z$ in $T$ such that $z\neq x,y$ and it holds that $\alpha(A_T(s))\setminus Z\neq\emptyset$ and $\alpha(A_T(z))$ is an $(x,y)$-separator in $G$. 
In particular, $x$ and $y$ cannot be both in $N_G[V(C)]$. By symmetry, we can assume that $x\notin N_G[V(C)]$. This means that there is a red-yellow component $H$ of $F$ such that   
 properties (i)--(iv) of step 5 are fulfilled. Because $X$ is a colorful solution, we have that $S=N_F[V(H)]\cap (B\cup Y)\subseteq X$. Thus, (a$'$) and (b$'$) are fulfilled for $Z'=Z\cup S$ and $h'=h-|S|$. As the subroutine calls $\textsc{FindX}(Z',h')$, we conclude that the claim if fulfilled. 

\medskip
Recall that we call $\textsc{FindX}(\emptyset,p+k)$ and note that conditions (a) and (b) are trivially fulfilled for $Z=\emptyset$ and $h=t$. Observe also that in each recursive call of the subroutine the parameter $h$ strictly decreases. Thus, we conclude that we output $X$ is some recursive call of $\textsc{FindX}(Z,h)$.
\end{proof}

Lemma~\ref{lem_correct-three} concludes the description of the algorithm and its correctness proof. We summarize and evaluate the running time in the following lemma. 

\begin{lemma}\label{lem_alg-three}  
\probED{\edphthree} on $(p,k)$-unbreakable graphs for $\varphi\in\Sigma_3$ can be solved in $2^{\Oh((p+k)(\log (p+k)+p))}\cdot n^{\Oh(|\varphi|)}$ time.
\end{lemma} 
 
\begin{proof} 
If $|V(G)|\leq (3p+2k)(p+1)$,  the problem is solved by  brute force  in $2^{(3p+2k)(p+1)}\cdot  ((3p+2k)(p+1))^{\Oh(k+|\varphi|)}$ time. Assume that  $|V(G)|> (3p+2k)(p+1)$.
Then we construct $\mathcal{F}$ in $2^{\Oh((p+k)\log(p+k))}\cdot n\log n$ time. The size of $\mathcal{F}$ is at most $2^{\Oh((p+k)\log(p+k))}\cdot \log n$, and
for every $R\in\mathcal{F}$, we construct $\mathcal{F}'$ in $2^{\Oh(\min\{p,k\}\log(p+k))}\cdot n\log n$ time. Recall that the size of $\mathcal{F}'$ is at most $2^{\Oh(\min\{p,k\}\log(p+k))}\cdot \log n$. Then we consider at most $n^r$ $r$-tuples of vertices $\mathbf{v}$ that can be listed in $n^{\Oh(|\varphi|)}$ time.
Finally, for every $R\in\mathcal{F}$,  $Y\in \mathcal{F}'$, and every $\mathbf{v}$, we call $\textsc{FindX}(\emptyset,p+k)$. 

Thus, it remains to evaluate the running time of $\textsc{FindX}(\emptyset,p+k)$.
Notice that in each call, $|Z|\leq p+k$. Then we can verify in $(p+k)^{\Oh(k)}\cdot n^{\Oh(1)}$ time whether $\dpt(Z)\leq k-1$ using Lemma~\ref{lem_repr-find}. Also we can check whether $(F,\mathbf{v})\models\varphi[\mathbf{x}]$ in $n^{\Oh(|\varphi|)}$ by Observation~\ref{obs:mod-check}. Simultaneously, we find an $s$-tuple $\mathbf{u}$ of vertices of $F$ such that $(F,\mathbf{vu})\not\models\varphi[\mathbf{xy}]$ if this is not the case. 
In step~4, we perform at most $s$ recursive calls. In step~5, finding $H$ can be done in polynomial time. Notice that we have at most $|Z|\leq p+k$ recursive calls in this step. The depth if the recursion is upper bounded by $k+p$. This implies that the running time of $\textsc{FindX}(\emptyset,p+k)$ is $(p+k)^{\Oh(p+k)}\cdot n^{|\varphi|}$.

Summarizing, we obtain that the total running time is $2^{\Oh((p+k)(\log (p+k)+p))}\cdot n^{\Oh(|\varphi|)}$.
\end{proof}

\section{Lower bound for $\Pi_3$-formulas}\label{sec_hard}
In this section, we complement Theorem~\ref{thm_algorithm_informal}  by proving that there are formulas in $\Pi_3$ for which \probED{\star} is \classW{2}-hard. We state now Theorem~\ref{thm_W_hard_informal} formally. 

\medskip\noindent\textbf{Theorem~2.}
\emph{For every $\star\in\{\edphione,\edphitwo,\edphthree\}$, there is $\varphi\in\Pi_3$ such that \probED{\star} is \classW{2}-hard.}
\medskip

\begin{proof}
We show that the problems are \classW{2}-hard for the formula $\varphi$ expressing the property that for every vertex $u$ of a graph, there is a vertex $v$ at distance at most two from $u$ of degree at most one. Notice that the property that a vertex $v$ of a given graph $G$ has degree at most one can be written as follows: for every $z_1,z_2\in V(G)$, if $v$ is adjacent to $z_1$ and $z_2$, then $z_1=z_2$.
 Thus, we define the formula 
\begin{equation*}
\psi(v,z_1,z_2)=[((v\sim z_1)\wedge(v\sim z_2))\rightarrow (z_1=z_2)]
\end{equation*}
with three free variables and set 
\begin{align*}
\varphi=\forall x\exists y_1\exists y_2 \forall z_1\forall z_2~&[\psi(x,z_1,z_2)\vee((x\sim y_1)\wedge\psi(y_1,z_1,z_2))\\
&\vee((x\sim y_1)\wedge(y_1\sim y_2)\wedge\psi(y_2,z_1,z_2))].
\end{align*}  
Clearly, $\varphi\in\Pi_3$.

To show \classW{2}-hardness, we reduce from the \textsc{Set Cover} problem. The problem asks, given a universe $U$, a family $\mathcal{S}$ of subsets of $\mathcal{S}$, and a positive integer $k$, whether there is $\mathcal{S}'\subset \mathcal{S}$ of size at most $k$ that covers $U$, that is, for every $u\in U$, there is $S\in\mathcal{S}'$ such that $u\in S$.
It is well-known that \textsc{Set Cover} is \classW{2}-complete when parameterized by $k$~\cite{DowneyF13}. 

\begin{figure}[ht]
\begin{center}
\scalebox{0.7}{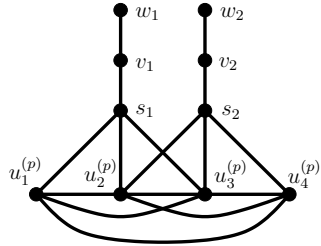}
\end{center}
\caption{Construction of $G$ for $n=2$ and $m=2$ with $S_1=\{u_1,u_2,u_3\}$ and $S_2=\{u_2,u_3,u_4\}$; for simplicity, just one copy of each $u_i^{(p)}$ for $p\in\{0,\ldots,k\}$ is shown.}\label{fig:hard}
\end{figure}

Let $(U,\mathcal{S},k)$ be an instance of \textsc{Set Cover} with $U=\{u_1,\ldots,u_n\}$, $\mathcal{S}=\{S_1,\ldots,S_m\}$. 
We also assume that $n\geq 2$ and $k\leq m$.
We construct the following graph $G$ (see Figure~\ref{fig:hard}).
\begin{itemize}
\item For every $i\in\{1,\ldots,n\}$, construct $k+2$ vertices $u_i^{(1)},\ldots,u_i^{(k+2)}$, and then for every $i,j\in\{1,\ldots,n\}$ and all $p,q\in \{1,\ldots,k+2\}$ such that $(i,p)\neq(j,q)$, make $u_i^{(p)}$ and $u_j^{(q)}$ adjacent.
\item For every $j\in\{1,\dots,m\}$, construct three vertices $s_j,v_j,w_j$ and edges $s_jv_j$ and $v_jw_j$.
\item For every $i\in\{1,\ldots,n\}$ and every $j\in\{1,\ldots,m\}$, make $s_j$ adjacent to $u_i^{(1)},\ldots,u_i^{(k+2)}$ if $u_i\in S_j$. 
\end{itemize}
 
We claim that $G$ has a set cover of size at most $k$ if and only if $\edphi{\star}(G)\leq k$ for $\star\in\{\edphione,\edphitwo,\edphthree\}$. Notice that by the definition of $\varphi$, $H\models\varphi$ if and only if $C\models\varphi$ for every component $C$ of $H$. Therefore, $\edphi{\edphione}(G)=\edphi{\edphitwo}(G)=\edphi{\edphthree}(G)$, and it is sufficient 
to prove that $G$ has a set cover of size at most $k$ if and only if there is an elimination set $X$ of $G$ with $\dpt(G)\leq k-1$ such that $C\models\varphi$ for every component $C$ of $G-X$.

Suppose that sets $S_{j_1},\ldots,S_{j_k}\in\mathcal{S}$ form a set cover. We define $X=\{w_{j_1},\ldots,w_{j_k}\}$.  Since $|X|=k$, $\dpt(X)\leq k-1$. Notice that $H=G-X$ is connected. We claim that $H\models \varphi$. Recall that  $H\models \varphi$ if and only if for every vertex $x\in V(H)$ there is a vertex $y\in V(H)$  at distance at most two such that $d_H(y)\leq 1$.     
 This property trivially holds if $x\in \{s_j,v_j,w_j\}\setminus X$ for $j\in \{1,\ldots,m\}$. Consider a vertex $u_i^{(p)}$ for some $i\in\{1,\ldots,n\}$ and $p\in\{1,\ldots,k+2\}$. We have that there is $h\in\{1,\ldots,k\}$ such that $u_i\in S_{j_h}$. Then $u_is_{j_h}\in E(H)$. Because $w_{j_h}\in X$, we obtain that $d_H(v_{j_h})=1$. Since $s_{j_h}v_{j_h}\in E(H)$, $v_{j_h}$ is at distance two from $u_i^{(p)}$ as required. We conclude that $H\models \varphi$. 
 
 For the opposite direction, assume that  there is an elimination set $X$ of $G$ with $\dpt(X)\leq k-1$ such that $C\models\varphi$ for every component $C$ of $G-X$. 
 Consider $Z=\{u_i^{(p)}\mid 1\leq i\leq n,~1\leq p\leq k+2\}$. Because $Z$ is a clique, we have that $|Z\cap X|\leq k$. To see this, consider a representation  $(T,\alpha)$ of $X$ with $\dpt(T)\leq k-1$. Then there is a leaf $x$ of $T$ such that $\alpha^{-1}(X\cap Z)\subseteq A_T(x)$. Since $\dpt(T)\leq k-1$, we conclude that $|Z\cap X|\leq k$. Note that the vertices of  $Z\setminus X$ are in the same component $H$ of $G-X$.     
Let $W=\{w_1,\ldots,w_m\}$.  By Observation~\ref{obs:anchor}, $|N_G(V(H))\cap X|\leq k$. Hence, $|N_G(V(H))\cap (X\cap W)|\leq k$ as well. 
Let $\{w_{j_1},\ldots,w_{j_\ell}\}=N_G(V(H))\cap (X\cap W)$. We claim that the sets $S_{j_1},\ldots,S_{j_k}$ cover $U$.

Consider an arbitrary $i\in\{1,\ldots,n\}$. Because $|Z\cap W|\leq k$, there are two distinct $p,q\in\{1,\ldots,k+2\}$ such that $u_i^{(p)},u_i^{(q)}\in V(H)$.  Since $H\models\varphi$, there is a vertex $z\in V(H)$ at distance at most two from $u_i^{(p)}$ such that $d_H(z)\leq 1$. Since $n\geq 2$, we have that $|Z\setminus X|\geq 3$ and, therefore,  
$d_H(u_i^{(p)})\geq 2$. Moreover, for every $h\in \{1,\ldots,n\}$ and $r\in\{1,\ldots,k+2\}$, if $u_h^{(r)}\in V(H)$, then $d_H(u_h^{(r)})\geq 2$. Then the construction of $G$ implies that there is $h\in\{1,\ldots,m\}$ such that $s_h\in V(H)$ and $u_i^{(p)}s_h\in E(G)$ with the property that either $d_H(s_j)\leq 1$ or $s_j$ has a neighboor in $H$ of degree at most one. As $s_h$ is adjacent to $u_i^{(p)}$, this vertex is adjacent to $u_i^{(q)}$. Hence, $d_H(s_h)\geq 2$. We obtain that $v_h\in V(H)$ and $d_H(v_h)\leq 1$ by the construction of $G$. This means that $w_h\notin H$, that is, $w_h\in N_G(V(H))\cap (X\cap W)$. We conclude that there is $t\in\{1,\ldots,\ell\}$ such that $j_t=h$. Finally, because $u_i^{(p)}$ is adjacent to $s_{j_t}$, $u_i\in S_{j_t}$ and this concludes the proof. \end{proof}

\section{Discussion}\label{sec_concl}
We established a parameterized complexity dichotomy for the elimination problems whose aim is to satisfy a FOL formula $\varphi$ with respect to the quantification structure of the prefix. For this, we considered three variants of the elimination distance to the class of graphs modelling $\varphi$ and defined the \probED{\star} for $\star\in\{\edphione,\edphitwo,\edphthree\}$ corresponding to the considered type of  distance.
In Theorem~\ref{thm_algorithm_informal}, we proved that for every FOL formula $\varphi\in\Sigma_3$, \probED{\star} is \classFPT for $\star\in\{\edphione,\edphitwo,\edphthree\}$.
In Theorem~\ref{thm_W_hard_informal}, we showed that this result is tight in the sense that there are FOL formulas $\varphi\in\Pi_3$ such that these problems are \classW{2}-hard.

Notice that the above dichotomy is the same for all the considered variants of the elimination problems. Moreover, it coincides with the structural dichotomy obtained by  
for  \probDEL by Fomin, Golovach, and Thilikos in~\cite{FominGT20}. This leads to the following natural question: is there a FOL formula $\varphi$ such that the parameterized complexity of   \probED{\star} for $\star\in\{\edphione,\edphitwo,\edphthree\}$ and \probDEL differs? In particular, is there a formula $\varphi$ such that \probDEL is \classFPT but one of the problem  \probED{\star} for $\star\in\{\edphione,\edphitwo,\edphthree\}$ turns to be, say \classW{1} or \classW{2}-hard? Note that Lemma~\ref{lem_expr} holds for every FOL formula $\varphi$. Thus,  solving \probED{\star} for $\star\in\{\edphione,\edphitwo,\edphthree\}$ can be reduced to solving these problems on unbreakable graphs by Theorem~\ref{thm_meta}. Since  \probED{\star} for $\star\in\{\edphione,\edphitwo,\edphthree\}$ is somehow similar to \probDEL on unbreakable graphs, it may happen that \probED{\star} for $\star\in\{\edphione,\edphitwo,\edphthree\}$ are \classFPT whenever \probDEL is \classFPT. However proving this would demand applying different algorithmic tools as our techniques are tailored for $\varphi\in\Sigma_3$. Also it would be interesting to know whether there are FOL formulas such that  \probED{\star} for $\star\in\{\edphione,\edphitwo,\edphthree\}$ differ from the parameterized complexity viewpoint. 

In contrast with the same behaviour of the elimination and deletion problems with respect to the inclusion in \classFPT, we would like to point that they behave differently with respect to \emph{kernelization} (we refer to the books~\cite{CyganFKLMPPS15,FomnLSZ19} for the definition of the  notion). It was shown in~\cite{FominGT20} that \probDEL admits a polynomial kernel for $\varphi\in \Sigma_1\cup\Pi_1$ (in fact, \probDEL is polynomial for $\varphi\in\Sigma_1$) and there are formulas $\varphi\in \Pi_2$ and $\Sigma_2$ such that \probDEL has no polynomial kernel unless $\classNP\subseteq\classCoNP/{\sf poly}$.  
For the elimination problems, we can show the following lower bound.

\begin{proposition}\label{prop:nokern}
There are formulas $\varphi\in \Pi_1$ such that \probED{\star} do not admit polynomial kernels unless $\classNP\subseteq\classCoNP/{\sf poly}$ for $\star\in\{\edphione,\edphitwo,\edphthree\}$.
\end{proposition}

\begin{proof}
We show the claim for the formula $\varphi$ expressing the property that a graph has no triangles, that is, cycles of length three: 
\begin{align*}
\varphi=\forall x\forall y\forall z~[(x=y)\vee(y=z)\vee(x=z)\vee\neg(x\sim y)\vee\neg(y\sim z)\vee\neg(x\sim z)].
\end{align*}
It is straightforward to see that $G\models\varphi$ if and only if $G$ has no triangles. 

By the classical results of Lewis and Yannakakis~\cite{lewis1980nodedeletion}, \probDEL is \classNP-complete. Then it  is easy to observe that the problem remains \classNP on instances $(G,k)$, where $G$ is a $(k+1)$-connected graph. For example, we can reduce from \probDEL on general graphs. Let $G$ be an $n$-vertex graph. We assume that $k<n-1$ as otherwise the problem is trivial. 
We construct the graph $G'$ from $G$ by adding $k+1$ copies of the complete bipartite graph $K_{n,n}$ and making each vertex of one part of the vertex partition to a unique vertex of $G$. Clearly, $G'$ is $(k+1)$-connected and it is easy to see that $G-X$ is triangle-free if and only if $G'$ has no triangles for every $X\subseteq V(G')$.
 This proves the \classNP-hardness for \probDEL on $(k+1)$-connected graphs.  Then Observation~\ref{obs:deletion} implies that \probED{\star} is \classNP-complete for every $\star\in\{\edphione,\edphitwo,\edphthree\}$.

Let $(G_1,k),\ldots (G_t,k)$ be instances of \probED{\star} for some $\star\in\{\edphione,\edphitwo,\edphthree\}$. 
Let $G$ be the disjoint union of $G_1,\ldots,G_t$. Then for $\star\in\{\edphione,\edphitwo\}$, we have that $(G,k)$ is a yes-instance of \probED{\star} if and only if $(G_j,k)$ is a yes-instance of \probED{\star} for every $j\in\{1,\ldots,t\}$. 
Then by the result of Bodlaender, Jansen, and Kratsch~\cite{BodlaenderJK14} (see also \cite[Part~III]{FomnLSZ19} for the introduction to the technique), \probED{\star} does not admit a polynomial kernel unless $\classNP\subseteq\classCoNP/{\sf poly}$. For \probED{\edphthree}, consider $G'$ that is the disjoint union of $G$ and $K_{k+3}$. Clearly, $\edphi{\edphthree}(K_{k+2})=k+1$. Then by Lemma~\ref{lem_depth-disconn}, $(G',k+1)$ is a yes-instance of \probED{\edphthree} if and only if  $(G_j,k)$ is a yes-instance of \probED{\edphthree} for every $j\in\{1,\ldots,t\}$. This implies that \probED{\edphthree} has no polynomial kernel unless $\classNP\subseteq\classCoNP/{\sf poly}$.
\end{proof}

Notice that Proposition~\ref{prop:nokern} does no exclude existence of \emph{Turing kernels} (we again refer to~\cite{CyganFKLMPPS15,FomnLSZ19} for the definition of the  notion).  This makes it natural to ask whether \probED{\star} admit polynomial Turing kernels for $\varphi\in\Sigma_3$ for $\star\in\{\edphione,\edphitwo,\edphthree\}$.

We defined the depth of a set $X\subseteq V(G)$ using a representation. However, there is an equivalent definition that uses the notion of \emph{tree-depth} (see, e.g.,~\cite{BulianD16} for the definition). Let $G$ be a graph and let $X\subseteq V(G)$. We define the \emph{torso} of $X$ as the graph $H$ obtained from $G[X]$ by making every two vertices $u,v\in X$ adjacent if three is a component $C$ of $G-X$ such that $u,v\in N_G(V(C))$. Then the following property can be shown by the definition of tree-depth.

\begin{observation}\label{obs:tree-depth}
For a set $X\subseteq V(G)$ and an integer $k$, $\dpt(X)\leq k-1$ if and only if the tree-depth of the torso of $X$ is at most $k$.
\end{observation}

Thus, \probED{\edphthree} can be stated as follows: given a graph $G$ and a nonnegative integer $k$, is there $X\subseteq V(G)$ whose torso has the tree-depth at most $k$ such that $G-X\models\varphi$? In other words, we ask whether there is a set of vertices whose torso has bounded tree-depth such that the graph obtained by the deletion of this set models our formula. Then we can consider the variants of \probED{\edphthree} for other ``width-measures''.  For example, what can be said about parameterized complexity of the variant of \probED{\edphthree}, where the \emph{tree-width} (see, e.g,~\cite{CyganFKLMPPS15} for the defintion) of the torso of $X$ should be at most $k-1$?

Finally, we believe that it could be interesting to consider yet another variant of the elimination distance. Recall that in the definitions of $\edphi{\star}$ for $\star\in\{\edphione,\edphitwo,\edphthree\}$, we considered properties of the components. In particular,
\begin{equation*}
\edphi{\edphitwo}(G)=
\begin{cases}
0,&\mbox{if }G\models\varphi,\\
1+\min_{v\in V(G)}\edphi{\edphitwo}(G-v), &\mbox{if\! }G\not\models\varphi\text{\! and\! }G\text{\! is\! connected},\\
\max\{1,\max\{\edphi{\edphitwo}(C)\mid C\text{ is a component of }G\}\},&\mbox{otherwise}.
\end{cases}
\end{equation*}
However, we can consider \emph{unions} of components instead. We say that graphs $G_1,\ldots,G_s$ form a \emph{component-partition} of $G$ if every component of $G$ is a component of $G_i$ for some $i\in\{1,\ldots,s\}$ and $G$ is the disjoint union of $G_1,\ldots,G_s$. Then, we can define
 \begin{equation*}
\edphi{\edphfour}(G)=
\begin{cases}
0,&\mbox{if\! }G\models\varphi,\\
1+\min_{v\in V(G)}\edphi{\edphfour}(G-v), &\mbox{if\! }G\not\models\varphi\text{\! and\! }G\text{\! is\! connected},\\
\min\{\max\{1,\edphi{\edphfour}(G_1),\ldots,\edphi{\edphfour}(G_s)\}\mid G_1,\ldots,G_s &\\
\hskip17mm\text {is a component partition of }G \} &\mbox{otherwise}.
\end{cases}
\end{equation*}
Then we can define the respective \probED{\edphfour} and investigate its parameterized complexity depending of $\varphi$. Note that our approach for solving the elimination problems fails in this case. In particular, we cannot express the problem using MSOL.

\end{document}

%% file: Fig-W.pdf_t
\begin{picture}(0,0)%
\includegraphics{Fig-W.pdf}%
\end{picture}%
\setlength{\unitlength}{3947sp}%
\begingroup\makeatletter\ifx\SetFigFont\undefined%
\gdef\SetFigFont#1#2#3#4#5{%
  \reset@font\fontsize{#1}{#2pt}%
  \fontfamily{#3}\fontseries{#4}\fontshape{#5}%
  \selectfont}%
\fi\endgroup%
\begin{picture}(2578,2233)(1541,-2023)
\put(3436, 27){\makebox(0,0)[lb]{\smash{{\SetFigFont{12}{14.4}{\rmdefault}{\mddefault}{\updefault}{\color[rgb]{0,0,0}$w_2$}%
}}}}
\put(2228,-1446){\makebox(0,0)[lb]{\smash{{\SetFigFont{12}{14.4}{\rmdefault}{\mddefault}{\updefault}{\color[rgb]{0,0,0}$u_2^{(p)}$}%
}}}}
\put(3383,-1453){\makebox(0,0)[lb]{\smash{{\SetFigFont{12}{14.4}{\rmdefault}{\mddefault}{\updefault}{\color[rgb]{0,0,0}$u_3^{(p)}$}%
}}}}
\put(1556,-1426){\makebox(0,0)[lb]{\smash{{\SetFigFont{12}{14.4}{\rmdefault}{\mddefault}{\updefault}{\color[rgb]{0,0,0}$u_1^{(p)}$}%
}}}}
\put(4038,-1456){\makebox(0,0)[lb]{\smash{{\SetFigFont{12}{14.4}{\rmdefault}{\mddefault}{\updefault}{\color[rgb]{0,0,0}$u_4^{(p)}$}%
}}}}
\put(2683,-840){\makebox(0,0)[lb]{\smash{{\SetFigFont{12}{14.4}{\rmdefault}{\mddefault}{\updefault}{\color[rgb]{0,0,0}$s_1$}%
}}}}
\put(2683,-427){\makebox(0,0)[lb]{\smash{{\SetFigFont{12}{14.4}{\rmdefault}{\mddefault}{\updefault}{\color[rgb]{0,0,0}$v_1$}%
}}}}
\put(2683, 27){\makebox(0,0)[lb]{\smash{{\SetFigFont{12}{14.4}{\rmdefault}{\mddefault}{\updefault}{\color[rgb]{0,0,0}$w_1$}%
}}}}
\put(3449,-853){\makebox(0,0)[lb]{\smash{{\SetFigFont{12}{14.4}{\rmdefault}{\mddefault}{\updefault}{\color[rgb]{0,0,0}$s_2$}%
}}}}
\put(3423,-420){\makebox(0,0)[lb]{\smash{{\SetFigFont{12}{14.4}{\rmdefault}{\mddefault}{\updefault}{\color[rgb]{0,0,0}$v_2$}%
}}}}
\end{picture}%

%% file: Elimination-FOL_full.bbl
\begin{thebibliography}{10}

\bibitem{SridharanPASK21}
{\sc A.~Agrawal, L.~Kanesh, F.~Panolan, M.~S. Ramanujan, and S.~Saurabh}, {\em
  An {FPT} algorithm for elimination distance to bounded degree graphs}, in
  38th International Symposium on Theoretical Aspects of Computer Science,
  ({STACS}), vol.~187 of LIPIcs, Schloss Dagstuhl - Leibniz-Zentrum f{\"{u}}r
  Informatik, 2021, pp.~5:1--5:11.

\bibitem{agrawal_et_al:LIPIcs:2020:13304}
{\sc A.~Agrawal and M.~S. Ramanujan}, {\em On the parameterized complexity of
  clique elimination distance}, in 15th International Symposium on
  Parameterized and Exact Computation (IPEC), vol.~180 of Leibniz International
  Proceedings in Informatics (LIPIcs), Dagstuhl, Germany, 2020, Schloss
  Dagstuhl--Leibniz-Zentrum f{\"u}r Informatik, pp.~1:1--1:13.

\bibitem{BodlaenderJK14}
{\sc H.~L. Bodlaender, B.~M.~P. Jansen, and S.~Kratsch}, {\em Kernelization
  lower bounds by cross-composition}, {SIAM} J. Discrete Math., 28 (2014),
  pp.~277--305.

\bibitem{borger2001classical}
{\sc E.~B{\"o}rger, E.~Gr{\"a}del, and Y.~Gurevich}, {\em The classical
  decision problem}, Springer Science \& Business Media, 2001.

\bibitem{BulianD16}
{\sc J.~Bulian and A.~Dawar}, {\em Graph isomorphism parameterized by
  elimination distance to bounded degree}, Algorithmica, 75 (2016),
  pp.~363--382.

\bibitem{BulianD17}
\leavevmode\vrule height 2pt depth -1.6pt width 23pt, {\em Fixed-parameter
  tractable distances to sparse graph classes}, Algorithmica, 79 (2017),
  pp.~139--158.

\bibitem{CaiCC06}
{\sc L.~Cai, S.~M. Chan, and S.~O. Chan}, {\em Random separation: {A} new
  method for solving fixed-cardinality optimization problems}, in Proceedings
  of the 2nd International Workshop Parameterized and Exact Computation
  (IWPEC), vol.~4169 of Lecture Notes in Computer Science, Springer, 2006,
  pp.~239--250.

\bibitem{ChitnisCHPP16}
{\sc R.~Chitnis, M.~Cygan, M.~Hajiaghayi, M.~Pilipczuk, and M.~Pilipczuk}, {\em
  Designing {FPT} algorithms for cut problems using randomized contractions},
  {SIAM} J. Comput., 45 (2016), pp.~1171--1229.

\bibitem{CourcelleE12}
{\sc B.~Courcelle and J.~Engelfriet}, {\em Graph Structure and Monadic
  Second-Order Logic - {A} Language-Theoretic Approach}, vol.~138 of
  Encyclopedia of mathematics and its applications, Cambridge University Press,
  2012.

\bibitem{CyganFKLMPPS15}
{\sc M.~Cygan, F.~V. Fomin, L.~Kowalik, D.~Lokshtanov, D.~Marx, M.~Pilipczuk,
  M.~Pilipczuk, and S.~Saurabh}, {\em Parameterized Algorithms}, Springer,
  2015.

\bibitem{CyganLPPS19}
{\sc M.~Cygan, D.~Lokshtanov, M.~Pilipczuk, M.~Pilipczuk, and S.~Saurabh}, {\em
  Minimum bisection is fixed-parameter tractable}, {SIAM} J. Comput., 48
  (2019), pp.~417--450.

\bibitem{Diestel12}
{\sc R.~Diestel}, {\em Graph Theory, 4th Edition}, vol.~173 of Graduate texts
  in mathematics, Springer, 2012.

\bibitem{DowneyF13}
{\sc R.~G. Downey and M.~R. Fellows}, {\em Fundamentals of Parameterized
  Complexity}, Texts in Computer Science, Springer, 2013.

\bibitem{FlumG06}
{\sc J.~Flum and M.~Grohe}, {\em Parameterized Complexity Theory}, Texts in
  Theoretical Computer Science. An {EATCS} Series, Springer, 2006.

\bibitem{FominGT20}
{\sc F.~V. Fomin, P.~A. Golovach, and D.~M. Thilikos}, {\em On the
  parameterized complexity of graph modification to first-order logic
  properties}, Theory Comput. Syst., 64 (2020), pp.~251--271.

\bibitem{FomnLSZ19}
{\sc F.~V. Fomin, D.~Lokshtanov, S.~Saurabh, and M.~Zehavi}, {\em
  Kernelization}, Cambridge University Press, Cambridge, 2019.
\newblock Theory of parameterized preprocessing.

\bibitem{FrickG04}
{\sc M.~Frick and M.~Grohe}, {\em The complexity of first-order and monadic
  second-order logic revisited}, Ann. Pure Appl. Logic, 130 (2004), pp.~3--31.

\bibitem{GottlobKS04}
{\sc G.~Gottlob, P.~G. Kolaitis, and T.~Schwentick}, {\em Existential
  second-order logic over graphs: Charting the tractability frontier}, J.
  {ACM}, 51 (2004), pp.~312--362.

\bibitem{DBLP:journals/jacm/GottlobKS04}
\leavevmode\vrule height 2pt depth -1.6pt width 23pt, {\em Existential
  second-order logic over graphs: Charting the tractability frontier}, J.
  {ACM}, 51 (2004), pp.~312--362.

\bibitem{GuoHN04}
{\sc J.~Guo, F.~H{\"{u}}ffner, and R.~Niedermeier}, {\em A structural view on
  parameterizing problems: Distance from triviality}, in Proceedings of the 1st
  International Workshop Parameterized and Exact Computation (IWPEC), vol.~3162
  of Lecture Notes in Computer Science, Springer, 2004, pp.~162--173.

\bibitem{hols_et_al:LIPIcs:2020:11897}
{\sc E.-M.~C. Hols, S.~Kratsch, and A.~Pieterse}, {\em {Elimination Distances,
  Blocking Sets, and Kernels for Vertex Cover}}, in 37th International
  Symposium on Theoretical Aspects of Computer Science (STACS), vol.~154 of
  Leibniz International Proceedings in Informatics (LIPIcs), Dagstuhl, Germany,
  2020, Schloss Dagstuhl--Leibniz-Zentrum f{\"u}r Informatik, pp.~36:1--36:14.

\bibitem{lewis1980nodedeletion}
{\sc J.~M. Lewis and M.~Yannakakis}, {\em The node-deletion problem for
  hereditary properties is {NP}-complete}, J. Comput. Syst. Sci., 20 (1980),
  pp.~219--230.

\bibitem{DLindermayrSV20}
{\sc A.~Lindermayr, S.~Siebertz, and A.~Vigny}, {\em Elimination distance to
  bounded degree on planar graphs}, in Proceeding of the 45th International
  Symposium on Mathematical Foundations of Computer Science (MFCS), vol.~170 of
  LIPIcs, Schloss Dagstuhl - Leibniz-Zentrum f{\"{u}}r Informatik, 2020,
  pp.~65:1--65:12.

\bibitem{LokshtanovR0Z18}
{\sc D.~Lokshtanov, M.~S. Ramanujan, S.~Saurabh, and M.~Zehavi}, {\em Reducing
  {CMSO} model checking to highly connected graphs}, in 45th International
  Colloquium on Automata, Languages, and Programming (ICALP), vol.~107 of
  LIPIcs, Schloss Dagstuhl - Leibniz-Zentrum f{\"{u}}r Informatik, 2018,
  pp.~135:1--135:14.

\bibitem{LokshtanovR0Z18a}
\leavevmode\vrule height 2pt depth -1.6pt width 23pt, {\em Reducing {CMSO}
  model checking to highly connected graphs}, CoRR, abs/1802.01453 (2018).

\bibitem{NesetrilO0806b}
{\sc J.~Ne{\v{s}}et{\v{r}}il and P.~Ossona~de Mendez}, {\em Tree-depth,
  subgraph coloring and homomorphism bounds}, European J. Combin., 27 (2006),
  pp.~1022--1041.

\bibitem{Niedermeierbook06}
{\sc R.~Niedermeier}, {\em Invitation to fixed-parameter algorithms}, vol.~31
  of Oxford Lecture Series in Mathematics and its Applications, Oxford
  University Press, 2006.

\bibitem{Smorynski77}
{\sc C.~Smorynski}, {\em The incompleteness theorems}, in Handbook of
  mathematical logic, vol.~90 of Stud. Logic Found. Math., North-Holland,
  Amsterdam, 1977, pp.~821--865.

\bibitem{Vardi82}
{\sc M.~Y. Vardi}, {\em The complexity of relational query languages (extended
  abstract)}, in Proceedings of the 14th Annual {ACM} Symposium on Theory of
  Computing, May 5-7, 1982, San Francisco, California, {USA}, {ACM}, 1982,
  pp.~137--146.

\bibitem{Williams14}
{\sc R.~Williams}, {\em Faster decision of first-order graph properties}, in
  Joint Meeting of the Twenty-Third {EACSL} Annual Conference on Computer
  Science Logic {(CSL)} and the Twenty-Ninth Annual {ACM/IEEE} Symposium on
  Logic in Computer Science {(LICS)}, {ACM}, 2014, pp.~80:1--80:6.

\end{thebibliography}
